%% file: paper.tex
\documentclass[sigconf,screen]{acmart}
\def\BibTeX{{\rm B\kern-.05em{\sc i\kern-.025em b}\kern-.08em
    T\kern-.1667em\lower.7ex\hbox{E}\kern-.125emX}}

\usepackage[mathscr]{eucal}
\usepackage{graphicx}
\usepackage{color}
\usepackage{todonotes}
\usepackage{algorithm,algpseudocode}
\usepackage{xspace}
\usepackage{etoolbox}
\usepackage{xcolor}
\usepackage{multirow}
\usepackage{enumitem}
\usepackage[normalem]{ulem}
\usepackage{bm}
\usepackage{subcaption}
\usepackage{booktabs}

\usepackage{mathtools}

\newtheorem{definition}{Definition}
\newtheorem{lemma}{Lemma}
\newtheorem{theorem}{Theorem}

\usepackage{tikz,pgfplots}
\usetikzlibrary{matrix, decorations, patterns, positioning, shapes, 3d, calc, intersections, arrows, fit, pgfplots.fillbetween}
\usepackage{tikz-3dplot}

\usepackage{cleveref}

\Crefname{definition}{Definition}{Definitions}
\crefname{definition}{Def.}{Defs.}
\Crefname{theorem}{Theorem}{Theorems}
\crefname{theorem}{Theorem}{Theorems}
\Crefname{lemma}{Lemma}{Lemmas}
\crefname{lemma}{Lemma}{Lemmas}
\Crefname{section}{Section}{Sections}
\crefname{section}{\S\hspace{-.5mm}}{\S}
\crefname{subsection}{\S\hspace{-.5mm}}{\S}
\crefname{algorithm}{Alg.}{Algs.}
\Crefname{algorithm}{Algorithm}{Algorithms}
\crefname{figure}{Fig.}{Figs.}
\Crefname{figure}{Figure}{Figures}
\crefname{table}{Tab.}{Tabs.}
\Crefname{table}{Table}{Tables}

\definecolor{pastelviolet}{rgb}{0.8, 0.6, 0.79}
\definecolor{babyblueeyes}{rgb}{0.63, 0.79, 0.95}
\definecolor{pastelyellow}{rgb}{0.99, 0.99, 0.59}
\definecolor{pastelgreen}{rgb}{0.47, 0.87, 0.47}
\definecolor{pastelred}{rgb}{1.0, 0.41, 0.38}
\colorlet{pattern blue}{blue!60}

\newcommand{\V}[2][]{{\bm{#1\mathbf{\MakeLowercase{#2}}}}} 
\newcommand{\M}[2][]{{\bm{#1\mathbf{\MakeUppercase{#2}}}}} 
\newcommand{\Tra}{{\sf T}}

\newcommand{\Real}{\mathbb{R}}
\newcommand{\vc}[1]{\bm{#1}}

\algrenewcommand{\algorithmicforall}{\textbf{for each}}
\algrenewcomment[1]{\hspace{\algorithmicindent} \(\triangleright\) \emph{#1}} 

\definecolor{pdfurlcolor}{rgb}{0,0,0.6}
\definecolor{pdfcitecolor}{rgb}{0,0.6,0}
\definecolor{pdflinkcolor}{rgb}{0.6,0,0}

\usepackage{pgfplots}
\usepackage{pgfplotstable}
\usepackage{tikz}
\usepackage{ifthen}
\usetikzlibrary{calc} 
\pgfplotsset{compat=1.16}

\definecolor{color3}{RGB}{70,130,180}   
\definecolor{color2}{RGB}{220,20,60}    
\definecolor{color1}{RGB}{34,139,34}    
\definecolor{color4}{RGB}{255,165,0}    
\definecolor{color5}{RGB}{128,128,128}  
\definecolor{color6}{RGB}{138,43,226}   
\definecolor{color7}{RGB}{255,105,180}  
\definecolor{color8}{RGB}{0,206,209}    
\definecolor{color9}{RGB}{210,105,30}   

\newcommand{\algoa}{1Dred}
\newcommand{\algob}{1Dnored}
\newcommand{\gen}{MatMul_Gen}
\newcommand{\comm}{MatMul_Comm}
\newcommand{\impla}{CPP}
\newcommand{\implb}{Python}
\newcommand{\implc}{SLATE}
\newcommand{\proca}{4}
\newcommand{\procb}{8}
\newcommand{\procc}{16}
\newcommand{\procd}{32}
\newcommand{\proce}{64}
\newcommand{\procf}{128}

\newif\ifylabel \ylabeltrue
\newif\iflegend \legendtrue

\newcommand{\matmulCPUrankFiveH}{
  ybar stacked,
  reverse legend,
  bar width=4pt,
  width=5.2cm, 
  height=4.5cm,
  enlarge x limits=0.05,    
  \ifylabel
  	ylabel={Time (seconds)}, 
  \fi
  y label style={yshift=-0.2cm},
  ymin=0,
  ymax=0.26,
  symbolic x coords={Algo-\gen-\proca,Algo-\comm-\proca,,Algo-\gen-\procb,Algo-\comm-\procb,,Algo-\gen-\procc,Algo-\comm-\procc,,Algo-\gen-\procd,Algo-\comm-\procd,,Algo-\gen-\proce,Algo-\comm-\proce,,Algo-\gen-\procf,Algo-\comm-\procf},
  xticklabels={Gen,Gen,Gen,Gen,Gen,Gen,Comm,Comm,Comm,Comm,Comm,Comm},  
  xtick=data,
  x tick label style={xshift=0.19cm, rotate=45, anchor=east}, 
  xlabel={P}, 
  xlabel style={yshift=-0.45cm}, 
  legend style={
        at={(0.65,1.1)},     
        anchor=west,         
    },
  legend columns=2,
  legend style={draw=none, cells={align=left}, nodes={scale=0.7}},
}

\newcommand{\matmulCPUrankFiveT}{
  ybar stacked,
  reverse legend,
  bar width=4pt,
  width=5.2cm, height=4.5cm,
  enlarge x limits=0.05,    
  \ifylabel
  \fi
  y label style={yshift=-0.2cm},
  ymin=0,
  ymax=.85,
  symbolic x coords={Algo-\gen-\proca,Algo-\comm-\proca,,Algo-\gen-\procb,Algo-\comm-\procb,,Algo-\gen-\procc,Algo-\comm-\procc,,Algo-\gen-\procd,Algo-\comm-\procd,,Algo-\gen-\proce,Algo-\comm-\proce,,Algo-\gen-\procf,Algo-\comm-\procf},
  xticklabels={Gen,Gen,Gen,Gen,Gen,Gen,Comm,Comm,Comm,Comm,Comm,Comm},  
  xtick=data,
  x tick label style={xshift=0.19cm, rotate=45, anchor=east}, 
  xlabel={P}, 
  xlabel style={yshift=-0.45cm}, 
  legend pos={north west},
  legend columns=2,
  legend style={draw=none, cells={align=left}, nodes={scale=0.7}},
}

\newcommand{\matmulGPUrankFiveH}{
  ybar stacked,
  reverse legend,
  bar width=4pt,
  width=5.2cm, height=4.5cm,
  enlarge x limits=0.05,    
  \ifylabel
  	ylabel={Time (seconds)}, 
  \fi
  y label style={yshift=-0.2cm},
  ymin=0,
  ymax=0.3,
  symbolic x coords={Algo-\gen-\proca,Algo-\comm-\proca,,Algo-\gen-\procb,Algo-\comm-\procb,,Algo-\gen-\procc,Algo-\comm-\procc,,Algo-\gen-\procd,Algo-\comm-\procd,,Algo-\gen-\proce,Algo-\comm-\proce,,Algo-\gen-\procf,Algo-\comm-\procf},
  xticklabels={Gen,Gen,Gen,Gen,Gen,Gen,Comm,Comm,Comm,Comm,Comm,Comm},  
  xtick=data,
  x tick label style={xshift=0.19cm, rotate=45, anchor=east}, 
  xlabel={P}, 
  xlabel style={yshift=-0.45cm}, 
    legend style={
        at={(0.65,1.1)},     
        anchor=west,         
    },
  legend columns=2,
  legend style={draw=none, cells={align=left}, nodes={scale=0.7}},
}

\newcommand{\matmulGPUrankFiveT}{
  ybar stacked,
  reverse legend,
  bar width=4pt,
  width=5.2cm, height=4.5cm,
  enlarge x limits=0.05,    
  \ifylabel
  \fi
  y label style={yshift=-0.2cm},
  ymin=0,
  ymax=0.4,
  symbolic x coords={Algo-\gen-\proca,Algo-\comm-\proca,,Algo-\gen-\procb,Algo-\comm-\procb,,Algo-\gen-\procc,Algo-\comm-\procc,,Algo-\gen-\procd,Algo-\comm-\procd,,Algo-\gen-\proce,Algo-\comm-\proce,,Algo-\gen-\procf,Algo-\comm-\procf},
  xticklabels={Gen,Gen,Gen,Gen,Gen,Gen,Comm,Comm,Comm,Comm,Comm,Comm},  
  xtick=data,
  x tick label style={xshift=0.19cm, rotate=45, anchor=east}, 
  xlabel={P}, 
  xlabel style={yshift=-0.45cm}, 
  legend pos={north west},
  legend columns=2,
  legend style={draw=none, cells={align=left}, nodes={scale=0.7}},
}

\newcommand{\nystromCPUrankFiveH}{
  ybar stacked,
  reverse legend,
  bar width=4pt,
  width=5.2cm, height=4.5cm,
  enlarge x limits=0.05,    
  \ifylabel
  	ylabel={Time (seconds)}, 
  \fi
  y label style={yshift=-0.2cm},
  ymin=0,
  ymax=0.035,
  symbolic x coords={Algo-\algoa-\proca,Algo-\algob-\proca,,Algo-\algoa-\procb,Algo-\algob-\procb,,Algo-\algoa-\procc,Algo-\algob-\procc,,Algo-\algoa-\procd,Algo-\algob-\procd,,Algo-\algoa-\proce,Algo-\algob-\proce,,Algo-\algoa-\procf,Algo-\algob-\procf},
  xticklabels={Redist,Redist,Redist,Redist,Redist,Redist,No-Redist,No-Redist,No-Redist,No-Redist,No-Redist,No-Redist},  
  xtick=data,
  x tick label style={xshift=0.1cm, rotate=65, anchor=east}, 
  xlabel={P}, 
  xlabel style={yshift=-0.15cm}, 
  legend style={
        at={(0.55,1.1)},     
        anchor=west,         
    },
  legend columns=3,
  legend style={draw=none, cells={align=left}, nodes={scale=0.7}},
}

\newcommand{\nystromCPUrankFiveT}{
  ybar stacked,
  reverse legend,
  bar width=4pt,
  width=5.2cm, height=4.5cm,
  enlarge x limits=0.05,    
  \ifylabel
  \fi
  y label style={yshift=-0.2cm},
  ymin=0,
  ymax=0.2,
  symbolic x coords={Algo-\algoa-\proca,Algo-\algob-\proca,,Algo-\algoa-\procb,Algo-\algob-\procb,,Algo-\algoa-\procc,Algo-\algob-\procc,,Algo-\algoa-\procd,Algo-\algob-\procd,,Algo-\algoa-\proce,Algo-\algob-\proce,,Algo-\algoa-\procf,Algo-\algob-\procf},
  xticklabels={Redist,Redist,Redist,Redist,Redist,Redist,No-Redist,No-Redist,No-Redist,No-Redist,No-Redist,No-Redist},  
  xtick=data,
  x tick label style={xshift=0.1cm, rotate=65, anchor=east}, 
  xlabel={P}, 
  xlabel style={yshift=-0.15cm}, 
  legend pos={north east},
  legend style={draw=none, cells={align=left}, nodes={scale=0.7}},
}

\newcommand{\nystromGPUrankFiveH}{
  ybar stacked,
  reverse legend,
  bar width=4pt,
   width=5.2cm, height=4.5cm,
  enlarge x limits=0.05,    
  \ifylabel
  	ylabel={Time (seconds)}, 
  \fi
  y label style={yshift=-0.2cm},
  ymin=0,
  ymax=0.15,
  symbolic x coords={Algo-\algoa-\proca,Algo-\algob-\proca,,Algo-\algoa-\procb,Algo-\algob-\procb,,Algo-\algoa-\procc,Algo-\algob-\procc,,Algo-\algoa-\procd,Algo-\algob-\procd,,Algo-\algoa-\proce,Algo-\algob-\proce,,Algo-\algoa-\procf,Algo-\algob-\procf},
  xticklabels={Redist,Redist,Redist,Redist,Redist,Redist,No-Redist,No-Redist,No-Redist,No-Redist,No-Redist,No-Redist},  
  xtick=data,
  x tick label style={xshift=0.1cm, rotate=65, anchor=east}, 
  xlabel={P}, 
  xlabel style={yshift=-0.45cm}, 
   legend style={
        at={(0.45,1.1)},     
        anchor=west,         
    },
  legend columns=3,
  legend style={draw=none, cells={align=left}, nodes={scale=0.7}},
}

\newcommand{\nystromGPUrankFiveT}{
  ybar stacked,
  reverse legend,
  bar width=4pt,
   width=5.2cm, height=4.5cm,
  enlarge x limits=0.05,    
  \ifylabel
  \fi
  y label style={yshift=-0.2cm},
  ymin=0,
  ymax=0.3,
  symbolic x coords={Algo-\algoa-\proca,Algo-\algob-\proca,,Algo-\algoa-\procb,Algo-\algob-\procb,,Algo-\algoa-\procc,Algo-\algob-\procc,,Algo-\algoa-\procd,Algo-\algob-\procd,,Algo-\algoa-\proce,Algo-\algob-\proce,,Algo-\algoa-\procf,Algo-\algob-\procf},
  xticklabels={Redist,Redist,Redist,Redist,Redist,Redist,No-Redist,No-Redist,No-Redist,No-Redist,No-Redist,No-Redist},  
  xtick=data,
  x tick label style={xshift=0.1cm, rotate=65, anchor=east}, 
  xlabel={P}, 
  xlabel style={yshift=-0.45cm}, 
  legend pos={north west},
  legend columns=1,
  legend style={draw=none, cells={align=left}, nodes={scale=0.7}},
}

\newcommand{\cppVSpython}{
  ybar stacked,
  reverse legend,
  bar width=5pt,
  width=8.4cm, height=6.5cm,
  enlarge x limits=0.05,    
  \ifylabel
  	ylabel={Time (seconds)}, 
  \fi
  y label style={yshift=-0.2cm},
  ymin=0,
  ymax=180,
  symbolic x coords={Algo-\impla-\proca,Algo-\implb-\proca,Algo-\implc-\proca,,Algo-\impla-\procb,Algo-\implb-\procb,Algo-\implc-\procb,,Algo-\impla-\procc,Algo-\implb-\procc,Algo-\implc-\procc,,Algo-\impla-\procd,Algo-\implb-\procd,Algo-\implc-\procd,,Algo-\impla-\proce,Algo-\implb-\proce,Algo-\implc-\proce,,Algo-\impla-\procf,Algo-\implb-\procf,Algo-\implc-\procf},
  xticklabels={CPP,CPP,CPP,CPP,CPP,CPP,Python,Python,Python,Python,Python,Python,SLATE,SLATE,SLATE,SLATE,SLATE,SLATE},  
  xtick=data,
  x tick label style={xshift=0.19cm, rotate=45, anchor=east}, 
  xlabel={P}, 
  xlabel style={yshift=-0.45cm}, 
  legend pos={north east},
  legend columns=2,
  legend style={draw=none, cells={align=left}, nodes={scale=0.7}},
}

\newcommand{\nystromCPUrankFiveTall}{
  ybar stacked,
  reverse legend,
  bar width=4pt,
  width=5.2cm, height=5.5cm,
  enlarge x limits=0.05,    
  \ifylabel
  	ylabel={Time (seconds)}, 
  \fi
  y label style={yshift=-0.2cm},
  ymin=0,
  ymax=21.0,
  symbolic x coords={Algo-\algoa-\proca,Algo-\algob-\proca,,Algo-\algoa-\procb,Algo-\algob-\procb,,Algo-\algoa-\procc,Algo-\algob-\procc,,Algo-\algoa-\procd,Algo-\algob-\procd,,Algo-\algoa-\proce,Algo-\algob-\proce,,Algo-\algoa-\procf,Algo-\algob-\procf},
  xticklabels={Redist,Redist,Redist,Redist,Redist,Redist,No-Redist,No-Redist,No-Redist,No-Redist,No-Redist,No-Redist},  
  xtick=data,
  x tick label style={xshift=0.1cm, rotate=65, anchor=east}, 
  xlabel={P}, 
  xlabel style={yshift=-0.15cm}, 
  legend style={
        at={(-0.18,1.1)},     
        anchor=west,         
    },
  legend columns=6,
  legend style={draw=none, cells={align=left}, nodes={scale=0.7}},
}

\newcommand{\nystromGPUrankFiveTall}{
  ybar stacked,
  reverse legend,
  bar width=4pt,
  width=5.2cm, height=5.5cm,
  enlarge x limits=0.05,    
  \ifylabel
  \fi
  y label style={yshift=-0.2cm},
  ymin=0,
  ymax=1.3,
  symbolic x coords={Algo-\algoa-\proca,Algo-\algob-\proca,,Algo-\algoa-\procb,Algo-\algob-\procb,,Algo-\algoa-\procc,Algo-\algob-\procc,,Algo-\algoa-\procd,Algo-\algob-\procd,,Algo-\algoa-\proce,Algo-\algob-\proce,,Algo-\algoa-\procf,Algo-\algob-\procf},
  xticklabels={Redist,Redist,Redist,Redist,Redist,Redist,No-Redist,No-Redist,No-Redist,No-Redist,No-Redist,No-Redist},  
  xtick=data,
  x tick label style={xshift=0.1cm, rotate=65, anchor=east}, 
  xlabel={P}, 
  xlabel style={yshift=-0.15cm}, 
  legend style={
        at={(0.05,1.1)},     
        anchor=west,         
    },
  legend columns=5,
  legend style={draw=none, cells={align=left}, nodes={scale=0.7}},
}

\usepackage{etoolbox}
\makeatletter
\patchcmd{\maketitle}{\@copyrightspace}{}{}{}
\makeatother

\begin{document}

\title{Communication Lower Bounds and Algorithms for Sketching with Random Dense Matrices}

\acmYear{2026}\copyrightyear{2026}
\setcopyright{cc}
\setcctype[4.0]{by}
\acmConference[SPAA '26]{38th ACM Symposium on Parallelism in Algorithms and Architectures}{July 6--10, 2026}{London, United Kingdom}
\acmBooktitle{38th ACM Symposium on Parallelism in Algorithms and Architectures (SPAA '26), July 6--10, 2026, London, United Kingdom}
\acmDOI{10.1145/3816782.3819223}
\acmISBN{979-8-4007-2761-0/26/07}

\author{Hussam Al~Daas}
\orcid{0000-0001-9355-4042}
\affiliation{%
	\institution{Rutherford Appleton Laboratory}
	\city{Didcot}
	\state{Oxfordshire}
	\country{UK}
}
\email{hussam.al-daas@stfc.ac.uk}

\author{Grey Ballard}
\orcid{0000-0003-1557-8027}
\affiliation{%
	\institution{Wake Forest University}
	\city{Winston-Salem}
	\state{NC}
	\country{USA}}
\email{ballard@wfu.edu}

\author{Laura Grigori}
\orcid{0000-0002-5880-1076}
\affiliation{%
	\institution{EPFL and PSI}
	\city{Lausanne}
	\country{Switzerland}}
\email{laura.grigori@epfl.ch}

\author{Md Taufique Hussain}
\orcid{0000-0001-9768-0564}
\affiliation{%
	\institution{Wake Forest University}
	\city{Winston-Salem}
	\state{NC}
	\country{USA}}
\email{hussaint@wfu.edu}

\author{Suraj Kumar}
\orcid{0009-0001-1449-0165}
\affiliation{%
	\institution{Inria Lyon}
	\city{Lyon}
	\country{France}
}
\email{suraj.kumar@inria.fr}

\author{Mohammad Marufur Rahman}
\orcid{0000-0002-5355-482X}
\affiliation{%
	\institution{Wake Forest University}
	\city{Winston-Salem}
	\state{NC}
	\country{USA}}
\email{rahmm224@wfu.edu}

\author{Kathryn Rouse}
\orcid{0009-0001-4045-0423}
\affiliation{%
	\institution{Inmar Intelligence}
	\city{Winston-Salem}
	\state{NC}
	\country{USA}
}
\email{kathryn.rouse@inmar.com}

\renewcommand{\shortauthors}{H. Al Daas, G. Ballard, L. Grigori, M. T. Hussain, S. Kumar, M. M. Rahman, and K. Rouse}
\begin{abstract}
    \input{abstract.tex}

\end{abstract}

\ccsdesc[500]{Theory of computation~Parallel algorithms}

\keywords{Communication costs, Matrix multiplications, Parallel algorithms, Randomized linear algebra, Nystr\"{o}m approximation}

\maketitle
\section{Introduction}
\label{sec:intro}

Many important linear algebra problems, such as computing low-rank
approximations of a matrix, selecting a subset of its columns, solving
linear systems, or computing eigenvalue decompositions, can be
efficiently solved using randomization techniques that allow reducing
computational and communication costs while providing accuracy
guarantees with high probability, see~\cite{cortinovis2026adaptive, bucci2025numerical, BalG22, frangella2023randomized}. These techniques rely on
\emph{sketching}, a dimensionality reduction approach in which a
matrix $\M{A} \in \mathbb{R}^{n_1 \times n_2}$ is mapped to a
lower-dimensional representation by multiplying it with a random
matrix $\M{\Omega} \in \mathbb{R}^{n_2 \times r}$, where generally $r \ll n_2$.
Depending on the application, the matrix $\M{A}$ may have different
dimensions and aspect ratios.  Different random matrices can be used
for sketching, including Gaussian matrices, structured fast
transforms, and sparse random matrices. In this work, we focus on
dense random matrices including Gaussian random matrices, which offer 
optimal theoretical guarantees for constructing low-dimensional 
representations. We refer the reader to \cite{MarT20} for a 
comprehensive overview of these sketching methods and their associated 
properties.

The goal of this work is to design efficient parallel algorithms for computing $\M{B} = \M{A} \M{\Omega}$, with a particular focus on avoiding unnecessary communication of random matrices. Because each processor can redundantly generate its needed entries of the random matrix, the product can be computed without communicating the random matrix, decreasing the amount of communication required from that of an explicit matrix multiplication (see, e.g., \cite{ABGKR22}). When $\M{A}$ has more rows than the number of processors, the communication lower bound is zero. Our optimal algorithm attains this bound by partitioning the rows of $\M{A}$ and $\M{B}$. The remaining cases are nontrivial, and the paper provides a detailed analysis of all cases. Our analysis also applies to $\M{B}^T =\M{\Omega}^T\M{A}^T$ computation by transposition. When the number of processors is small, the communication lower bound for explicit matrix multiplication is approximately the size of the smallest matrix.

We then focus on an application of sketching to the Nystr\"{o}m approximation that computes a low rank approximation of a symmetric square matrix $\M{A}$ as $\widetilde{\M{A}} = (\M{A}\M{\Omega}) (\M{\Omega}^T\M{A}\M{\Omega})^\dag (\M{A}\M{\Omega})^T$ \cite{Nys930,WilS00,GM-2016}, for some random matrix $\M{\Omega}$. 
This representation involves first computing $\M{B}=\M{A}\M{\Omega}$ and then computing $\M{C} = \M{\Omega}^T \M{B}$.

We assume that processors perform an equal amount of computation and that there is one copy of the input data at the beginning and one copy of the output data at the end. The lower bounds do not make any assumptions regarding data distribution. Our lower bound approach uses a geometric inequality that relates computation to data to build a constrained optimization problem whose analytical solution yields communication lower bounds. Our algorithms select data distributions that minimize communication. We consider both a single matrix multiplication involving a random input as well as the entire Nystr\"{o}m computation.

We address both theoretical and practical aspects by establishing communication lower bounds and analyzing asymptotic algorithmic costs, as well as by implementing the most communication efficient algorithms and applying them to large data sets on 100s of nodes (100s of GPUs or 1000s of CPU cores) of a supercomputer.
We develop implementations using both Python and C++ and adapt them for both CPU-only and GPU platforms and explore the performance tradeoffs among the various configurations.
Our results show the benefits of GPUs for performing dense matrix multiplication to accelerate Nystr\"{o}m approximation and the effectiveness of direct communication among GPUs to scale to large problems, and we obtain low-rank approximations of symmetric matrices of dimension 50,000 in fractions of a second.
Our code is publicly available at \url{https://github.com/taufique71/parallel-nystrom}.

In this work, we focus on $\M{A}\M{\Omega}$ and $\M{\Omega}^T\M{A}\M{\Omega}$ computations for dense and random $\M{\Omega}$.
The main contributions are to
\begin{enumerate}[leftmargin=1.5em]
	\item establish communication lower bounds for the parallel computation of $\M{B} = \M{A}\M{\Omega}$ where $\M{\Omega}$ is a dense random matrix and present parallel algorithms whose communication costs are optimal in all ranges (see \cref{sec:randMatmul});
	\item establish communication lower bounds for the parallel computation of the sequence of computations $\M{B} = \M{A}\M{\Omega}$ followed by $\M{C}=\M{\Omega}^T\M{B}$ where $\M{\Omega}$ is a dense random matrix, and present parallel algorithms whose communication costs are close to the lower bounds (see \cref{sec:nystrom});
	\item implement the most efficient algorithms using both Python and C++ for CPUs and GPUs (code available at \url{https://github.com/taufique71/parallel-nystrom}) and benchmark the implementations to demonstrate the communication efficiency of the algorithms and their parallel scaling (see \cref{sec:experiments}).
\end{enumerate}

\section{Related Work}
\label{sec:related}
There have been numerous studies on establishing communication lower bounds for matrix multiplication. 
The first bounds were proposed by Hong and Kung~\cite{HK81} for several computations including matrix multiplication. 
Aggarwal et al.~\cite{ACS90} extended this work for the parallel model and established first memory-independent parallel communication lower bounds. 
Irony et al.~\cite{IRONY:JPDC04} reproduced the parallel matrix multiplication bounds using a geometric inequality. 
Al Daas et al.~\cite{ABGKR22} established tight parallel communication lower bounds for all ranges of matrix dimensions and numbers of processors.
Liang et al.~\cite{TMBD-2024} presented a theoretical analysis to determine the data transfer cost of their algorithm for multiplying a dense random matrix with a sparse matrix. 
Their algorithm performs less data movement than the classical GEMM bound.

Balabanov et al.~\cite{BalBGL23} proposed a parallel sketching using a variant of the subsampled randomized Hadamard transform that exploits the structure of the random matrix.
Talwalkar et al.~\cite{talwalkar2013large} presented a parallel implementation of sketching and Nystr\"om approximation applied to kernel matrices with very large datasets. Higgins et al.~\cite{HigBY25} presented a high performance GPU implementation of sketching using count sketch and its application to least squares problems. Chen et al.~\cite{CheNRSPK25} presented a parallel (multi GPU) implementation and a benchmark study of sparse sign sketching and its application to solve least squares problems.

Nystr\"{o}m approximation has been considered for a wide range of random matrices.
One of the most widely used techniques is based on random sampling. 
That is, the columns in $\M{\Omega}$ are randomly selected from the identity matrix. 
The resulting Nystr\"{o}m approximation is given as $\widetilde{\M{A}} = \M{A}(:, c) \M{A}(c,c)^\dag \M{A}(:,c)^T$, where $c$ is the randomly selected subset of the column indices. 
This approximation is very easy to compute as it only requires access to a subset of the matrix entries and does not involve any matrix multiplication. 
However, the approximation quality can be very poor if $c$ is not chosen carefully. Several sampling techniques relying on sketching have been investigated in the literature and applied to Nystr\"om approximation~\cite{cortinovis2026adaptive, cortinovis2025sublinear, bucci2025numerical, park2025accuracy}. 
Gittens and Mahoney~\cite{GM-2016} studied the performance quality and runtime of high-quality random sampling and random projection methods for Nystr\"{o}m approximation on various symmetric positive semi-definite matrices using MATLAB. They observed that both methods exhibit similar running times, while the relative performance quality depends on the specific parameter of interest, with no clear overall winner. Li et al.~\cite{LLW-2023} studied the optimal convergence rates of Nystr\"{o}m approximation under sampling-based selection in a distributed environment.

\section{Model and Preliminaries}
\label{sec:prelim}
Multiplication with a random matrix in linear algebra is performed to achieve dimension reduction. 
Thus the non-contracted dimension, $r$, of the random matrix is much smaller than the contracted dimension $n$. 
Therefore, we assume that $r<n$ throughout this paper. 
Additionally, we assume that all matrix multiplications are performed using a classical matrix multiplication algorithm.

\subsection{Parallel Computation Model}
\label{sec:prelim:model}

We first present the MPI model of computation and communication, which we use to establish communication lower bounds and analyze the costs of our algorithms in \cref{sec:randMatmul,sec:nystrom}.
In this model, computation is distributed across $P$ processors.
Each processor has a private local memory and is connected to all others via a fully connected network. 
Processors can only operate on data residing in their local memory and must communicate to access data stored in other processors.
We assume that each processor can send and receive at most one message at a time. 
Communication refers to send and receive operations that transfer data between the private local memory of a processor and the network. 
The cost of communication mainly depends on the total volume of data transferred (bandwidth cost) and the number of messages exchanged (latency cost). 
In our model, the communication cost of an algorithm is defined as the maximum communication cost among all paths in the corresponding computation dependency graph.
For large messages, bandwidth cost typically dominates latency, so we focus mainly on bandwidth cost.

Our algorithms in \cref{sec:randMatmul,sec:nystrom,sec:experiments} use All-Gather, Reduce-Scatter, and All-to-All collectives.  
The optimal latency and bandwidth costs of All-Gather and Reduce-Scatter collectives on $Q$ processors are $\log(Q)$ and $(1-\frac{1}{Q})W$, respectively, where $W$ denotes the words of data in each processor after All-Gather or before Reduce-Scatter collective~\cite{Thakur:CollectiveCommunications:2005,Chan:CollectiveCommunications:2007}.
For All-to-All, we assume a bandwidth cost of $W$, where $W$ denotes the amount of data a processor starts and ends with, and a latency cost of $Q-1$ for $Q$ processors~\cite{BCKUW97}.

\subsection{Fundamental Existing Results}
\label{sec:prelim:existing}
We now present existing results which we will use to prove our lower bound results (\cref{thm:randMatmulLB,thm:NystromLB}). 
The first lemma relates the size of a set with its projections onto lower dimensional subspaces. This is a generalization of the Loomis-Whitney inequality~\cite{LW49} that applies to arbitrary subsets of projections.
We use it in \cref{thm:hbl:random} to relate the size of a $3$-dimensional set to its $1$ and $2$-dimensional projections. 
In our context, the $3$-dimensional set contains iteration points corresponding to the computation and the $2$-dimensional set contains indices of array accesses.
The lemma is proved in \cite{Christ:EECS-2013-61} but we use the statement from~\cite[Lemma 4.1]{Ballard:MTTKRP:IPDPS18}.
\begin{lemma}
	\label{lem:hbl}
	Consider any positive integers $\ell$ and $m$ and any $m$ projections $\phi_j:\mathbb{Z}^\ell\rightarrow\mathbb{Z}^{\ell_j}$ ($\ell_j\leq \ell$), each of which extracts $\ell_j$ coordinates $S_j\subseteq [\ell]$ and forgets the $\ell-\ell_j$ others.
	Define
	$\mathcal{C} = \big\{\V{s} \in[0,1]^m:\M{\Delta}\cdot\V{s}\ge\V{1}\big\}\text,$
	where the $\ell\times m$ matrix $\M{\Delta}$ has entries
	$\M{\Delta}_{i,j} = 1 \text{ if } i\in S_j \text{ and } \M{\Delta}_{i,j} = 0 \text{ otherwise}\text.$
	If $[s_1\ \cdots \ s_m]^\Tra\in\mathcal{C}$, then for all $F\subseteq \mathbb{Z}^\ell$,
	$$ |F| \leq \prod_{j\in [m]}|\phi_j(F)|^{s_j}\text.$$
\end{lemma}

The following lemma provides a bound on the minimum number of elements a processor must access from each matrix to perform at least $1/P$th of the scalar multiplications required for a classical matrix multiplication computation. 
We use it in the proofs of \cref{thm:randMatmulLB,thm:NystromLB} to obtain additional constraints for our optimization.
\begin{lemma}[{\cite[Lemma 4]{ABGKR22}}]
	\label{lem:proj}
	Given a parallel matrix multiplication algorithm that multiplies an $n_1\times n_2$ matrix $\M{A}$ by an $n_2\times n_3$ matrix $\M{B}$ using $P$ processors, any processor that performs at least $1/P$th of the scalar multiplications must access at least $n_1n_2/P$ elements of $\M{A}$ and at least $n_2n_3/P$ elements of $\M{B}$ and also compute contributions to at least $n_2n_3/P$ elements of $\M{C}=\M{A}\cdot \M{B}$.
\end{lemma}
The following definition and results allow us to solve the optimization problems used to obtain communication lower bounds in \cref{sec:randMatmul,sec:nystrom}.

\begin{definition}[{\cite[eq. (5.49)]{BV04}}]
	\label{def:KKT}
	Consider an optimization problem of the form 
	\begin{equation}
	\label{eq:optprob}
	\min_{\vc{x}} f(\vc{x}) \quad \text{ subject to } \quad \vc{g}(\vc{x}) \leq \vc{0}
	\end{equation} 
	where $f:\Real^d \rightarrow \Real$ and $\vc{g}:\Real^d\rightarrow \Real^c$ are both differentiable. 
	Define the dual variables $\vc{\mu}\in\mathbb{R}^c$, and let $\vc{J}_{\vc{g}}$ be the Jacobian of $\vc{g}$.
	The \emph{Karush-Kuhn-Tucker (KKT)} conditions of $(\vc{x},\vc{\mu})$ are as follows:
	\begin{itemize}
		\item \emph{Primal feasibility}: $\vc{g}(\vc{x}) \leq \vc{0}$;
		\item \emph{Dual feasibility}: $\vc{\mu} \geq 0$;
		\item \emph{Stationarity}: $\nabla f(\vc{x}) + \vc{\mu} \cdot \vc{J}_{\vc{g}}(\vc{x}) = \vc{0}$;
		\item \emph{Complementary slackness}: $\mu_i g_i(\vc{x})=0$ for all $i\in \{1,\dots,c\}$. 
	\end{itemize}
\end{definition}

\begin{lemma}[{\cite[Lemma 3]{ABGKR22}}]
	\label{lem:KKT}
	Consider an optimization problem of the form given in \cref{eq:optprob}.
	If $f$ is a convex function and each $g_i$ is a quasiconvex function, then the KKT conditions are sufficient for optimality.	
\end{lemma}

\begin{lemma}[{\cite[Lemma 2.2]{BR20}}]
	\label{lem:quasiconvex:xi}
	The function $g_0(\vc{x}) = L - \prod_{i=1}^{d} x_{i}$, for some constant $L$, is quasiconvex in the positive quadrant.
\end{lemma}

\subsection{Fundamental New Results}
\label{sec:prelim:new}

We now present a geometric inequality that we will use in \cref{sec:randMatmulLB,sec:nystrom:lb} to derive lower bounds for computations involving random matrices. 
\begin{theorem}
	\label{thm:hbl:random}
	Let $V$ be a finite set of points in $\mathbb{Z}^3$. Let $\phi_{ij}(V)$ be the projection of $V$ on the $ij$-plane, i.e., all points $(i,j)$ such that there exists a $k$ so that $(i,j,k) \in V$. Define $\phi_{jk}(V)$ and $\phi_{ki}(V)$ similarly. 
	Then $|V| \leq |\phi_{ij}(V)| \cdot |\phi_{jk}(V)|,$  $|V| \leq |\phi_{jk}(V)| \cdot |\phi_{ki}(V)|\text,$ and $|V| \leq |\phi_{ki}(V)| \cdot |\phi_{ij}(V)|\text.$
\end{theorem}
\begin{proof}
	To begin, define $\phi_i(V)$ to be the projection of $V$ on the $i$-axis, i.e., all points $i$ such that there exists a $(j,k)$ so that $(i,j,k) \in V$, and $\phi_j(V)$ and $\phi_k(V)$ similarly.
	We recall from \cref{lem:hbl} that $\Delta_{i,j}=1$ if index $i$ is present in projection $j$ and $\Delta_{i,j}=0$ otherwise.
	We see that $|V| \leq |\phi_i(V)| \cdot |\phi_{jk}(V)|$ 
	by setting
	$\Delta = \begin{bmatrix} 1 & 0 \\ 0 & 1 \\ 0 & 1 \end{bmatrix}$
	and $\V{S} = \begin{bmatrix} 1 & 1 \end{bmatrix}$, and then applying \cref{lem:hbl}. The first and second columns of $\Delta$ correspond to $\phi_i(V)$ and $\phi_{jk}(V)$ projections, respectively.
	Likewise, we obtain $|V| \leq |\phi_j(V)| \cdot |\phi_{ki}(V)|\text, $ and $|V| \leq |\phi_k(V)| \cdot |\phi_{ij}(V)|$.
	
	To prove the result, note that $i\in \phi_i(V)$ implies that there is a $(j,k)$ such that $(i,j,k)\in V$,  this also implies that $(i,j)\in \phi_{ij}(V)$.
	Thus $|\phi_i(V)|\leq |\phi_{ij}(V)|$ so the first inequality holds.
	A similar argument applies to the remaining inequalities.
\end{proof}

\section{Lower Bounds and Algorithm for Matrix Multiplication with a Random Matrix}
\label{sec:randMatmul}
In this section we consider the computation $\M{B} = \M{A}\M{\Omega}$ where $\M{A}$ is a matrix of dimensions $n_1\times n_2$ and $\M{\Omega}$ is a random matrix of dimensions $n_2\times r$ where $ r < n_2$.
The iteration space of the computation contains $n_1n_2r$ iteration points and a scalar multiplication is performed in each iteration point.
As $\M{\Omega}$ is a random matrix it does not need to be communicated but can instead be generated by any processor that needs it.
Thus we expect that an algorithm to compute this should require less communication than one that computes a classical matrix multiplication between non-random matrices.

\subsection{Communication Lower Bound}
\label{sec:randMatmulLB}
We begin by presenting an abstract optimization problem which we will use to determine the minimum amount of data required for a processor to perform its assigned scalar multiplications. 
In our lower bound formulation, $x_1$ and $x_2$ correspond to the number of elements of non-random matrices accessed by a processor, and we seek to minimize their sum to establish the bound.

\begin{lemma}
	\label{lem:randMatmulOpt}
	Consider the following optimization problem 
	$$\min_{\V{x}\in\mathbb{R}^2} x_1 + x_2$$
	such that
	$x_1x_2 \geq n_1n_2r/P,$ $n_1n_2/P\leq x_1,$ $n_1r/P\leq x_2,$
	where $P,n_1,n_2$, and $r$ are all positive integers greater than or equal to 1, and $n_2> r.$
	The constraints induce three cases for the optimal solution $\V{x}^*$ :
	\begin{itemize}
		\item If $P\leq n_1$ then $x_1^* = \frac{n_1n_2}{P}, x_2^* = \frac{n_1r}{P};$
		\item If $n_1 < P \leq n_1n_2/r$ then $x_1^* = \frac{n_1n_2}{P}, x_2^*=r;$
		\item If $n_1n_2/r < P$ then $x_1^*=x_2^*=\left(\frac{n_1n_2r}{P}\right)^{1/2}$.
	\end{itemize}
	This can be visualized as follows:
	\begin{center}
		\begin{tikzpicture}[scale=0.5, every node/.style={transform shape}]
		\draw [->, thick] (-0.1,0) -- (15,0) node [above,scale=2] {Increasing $P$};s
		\draw (0, 0.1) -- node [below, pastelred, scale=2]{$1$}(0,-0.1);
		\draw (5, 0.1) -- node [below, pastelred, scale=2]{$n_1$}(5,-0.1);
		\draw (10, 0.1) -- node [below, pastelred, scale=2] {$\frac{n_1n_2}{r}$}(10,-0.1);
		
		\node[align=left,below,scale=1.5] at (2.5, -0.4) {$x_1^*=\frac{n_1n_2}{P}$\\ $x_2^*=\frac{n_1r}{P}$};
		\node[align=left,below,scale=1.5] at (7.5, -0.6) {$x_1^*=\frac{n_1n_2}{P}$\\ $x_2^*=r$};
		\node[align=center,below,scale=1.5] at (13.25, -0.8) {$x_1^*=x_2^*= \left(\frac{n_1n_2r}{P}\right)^{1/2}$};	
		\end{tikzpicture}
		\end{center}
\end{lemma}

\begin{proof}
	By \cref{lem:KKT}, we can establish the optimality of the solution for every case by showing that the KKT conditions specified in \cref{def:KKT} are satisfied because the objective function and all but the first constraint are affine functions, which are convex (and quasiconvex), and the first constraint is quasiconvex in the positive quadrant by \cref{lem:quasiconvex:xi}.
	
	To match standard notation, let $f(\V{x}) = x_1+x_2$ and
	$$g(\V{x}) = \begin{bmatrix} \frac{n_1n_2r}{P}-x_1x_2 \\ \frac{n_1n_2}{P}-x_1 \\ \frac{n_1r}{P}-x_2\end{bmatrix}.$$
	Then $\nabla f(\V{x}) = \begin{bmatrix}1 & 1\end{bmatrix}$ and
	$$J_g(\V{x}) = \begin{bmatrix} 
	-x_2 & -x_1 \\
	-1 & 0 \\
	0 & -1
	\end{bmatrix}.
	$$
	
	In each of the three cases, we will set $\V{x}^*$ and $\V{\mu}^*$ then verify that the KKT conditions hold.
	\paragraph{Case 1 ($P\leq n_1$)}
	Set $\V{x}^* = \begin{bmatrix} n_1n_2/P & n_1r/P \end{bmatrix}$ and $\V{\mu}^* = \begin{bmatrix} 0  & 1 & 1  \end{bmatrix}$. 
	Both primal and dual feasibility are immediate. 
	As $\V{\mu}^* J_g(\V{x}^*) \linebreak[4] = \begin{bmatrix} -1 & -1\end{bmatrix},$ the stationarity condition holds.
	Complementary slackness is satisfied as all but the first constraint are tight for $\V{x}^*$ and the dual variable $\mu_1^*$ is zero.

	\paragraph{Case 2 ($n_1 < P \leq n_1n_2/r$)}
	Set $\V{x}^* = \begin{bmatrix} n_1n_2/P & r \end{bmatrix}$ and $\V{\mu}^* = \begin{bmatrix} P/(n_1n_2) & 1-(rP)/(n_1n_2) & 0\end{bmatrix}$. 
	The primal feasibility of $x_2$ is satisfied because $n_1 \leq P$ implies $ n_1r/P \leq r$. 
	The other constraints are clearly satisfied. 
	Dual feasibility requires that $1-(rP)/(n_1n_2) \geq 0$, which is satisfied from the condition of the case, $P \leq n_1n_2/r$.
	Stationarity can be directly verified. 
	Complementary slackness follows because the only constraint which is not tight is associated with a zero dual variable.
	
	\paragraph{Case 3 ($ n_1n_2/r < P$)}
	Set $\V{x}^* = \begin{bmatrix} (n_1n_2r/P)^{1/2} & (n_1n_2r/P)^{1/2} \end{bmatrix}$ and $\V{\mu}^* = \begin{bmatrix} (P/(n_1n_2r))^{1/2} & 0 & 0\end{bmatrix}$. 
	Primal feasibility is satisfied because $n_1n_2/r < P$ implies that $n_1n_2/P \leq (n_1n_2r/P)^{1/2}$.
	Dual feasibility is immediate and stationarity is directly verified. 
	Complementary slackness is satisfied because the first constraint is tight and only the first dual variable is not zero.
\end{proof}

\begin{theorem}
	\label{thm:randMatmulLB}
	Consider the computation, $\M{B}=\M{A}\cdot\M{\Omega}$, where $\M{A}$ has dimensions $n_1\times n_2$ and $\M{\Omega}$ is a random matrix of dimensions $n_2\times r$ where $n_2 > r$. 
	Any parallel algorithm using $P$ processors that load balances the computation and begins with one copy of the input matrix $\M{A}$ and ends with one copy of the output matrix $\M{B}$ must communicate at least $W$ words of data where 
	$$W= \begin{cases}
	\quad 0 &\text{ if } \quad 1\leq P \leq n_1\\
	\quad r-\frac{n_1r}{P} &\text{ if } \quad n_1 < P \leq \frac{n_1n_2}{r}\\
	\quad 2\left(\frac{n_1n_2r}{P}\right)^{1/2} - \frac{n_1n_2+n_1r}{P} &\text{ if } \quad \frac{n_1n_2}{r} < P\\
	\end{cases}.$$
\end{theorem}

\begin{proof}
	The total number of iteration points in the computation is $n_1n_2r$.
	As the algorithm load balances the computation, each processor performs the scalar multiplications associated with $n_1n_2r/P$ iteration points. 
	We focus on a processor for which the sum of the sizes of its input at the start and output at the end of the computation does not exceed $(n_1n_2 + n_1r)/P$ words. 
	Such a processor must exist as otherwise the algorithm would either start with more than one copy of $A$ or end with more than one copy of $B$. 	

	Let $F$ be the set iteration points corresponding to the scalar multiplications performed by this processor. 
	An element $(i,j,k)$ of $F$ corresponds to the multiplication of $\M{A}(i,j)$ with $\M{\Omega}(j,k)$, which contributes to the result $\M{B}(i,k)$. 
	We consider the projections of the iteration points in $F$ onto the indices of the arrays.
	Let $\phi_{ij}(F)$ and  $\phi_{ik}(F)$ denote the projections of $F$ onto the input matrix $\M{A}$ and the output matrix $\M{B}$ respectively. 
	Then \cref{thm:hbl:random} implies that  $|\phi_{ij}(F)||\phi_{ik}(F)| \geq |F| = n_1n_2r/P$. 	
	By \cref{lem:proj} we know that $|\phi_{ik}(F)| \geq n_1r/P$ and $|\phi_{ij}(F)|\geq n_1n_2/P$.  
	
	To minimize the communication we need to minimize the number of array elements accessed by this processor.
	As the processor must access all elements in $\phi_{ij}(F)$ and $\phi_{ik}(F)$ in order to perform all the scalar multiplications corresponding to the iteration points in $F$, we want to minimize $|\phi_{ij}(F)| + |\phi_{ik}(F)|$ subject to the above constraints. 
	By \cref{lem:randMatmulOpt}, we obtain the minimum number of elements that must be accessed by this processor. 
	Subtracting the data the processor can own from the result proves the lower bound. 
\end{proof}

Note that by replacing all the projections on matrices in the above proof with the projections on their transposes, we can show that the above lower bound is also valid for the computation, $\M{C}^T = \M{\Omega}^T\M{A}^T$. 
Hence the above theorem also allows one to obtain communication lower bounds for matrix multiplications where the random matrices are the first operands.

In most practical cases, the number of rows of $A$ is larger than $P$; the bound indicates that no communication is required. 
In general, the obtained bounds are better (smaller) than the general matrix multiplication bounds~\cite{ABGKR22}. This motivates us to design algorithms that achieve these bounds.
For $n_1 \geq n_2 > r$, the general matrix multiplication bound ($LB_{GEMM}$) is visualized as follows.
\begin{center}
	\begin{tikzpicture}[scale=0.5, every node/.style={transform shape}]
	\draw [->, thick] (-0.1,0) -- (15.5,0) node [above,scale=1.5] {Increasing $P$};s
	\draw (0, 0.1) -- node [below, pastelred, scale=2]{$1$}(0,-0.1);
	\draw (5, 0.1) -- node [below, pastelred, scale=2]{$\frac{n_1}{n_2}$}(5,-0.1);
	\draw (10, 0.1) -- node [below, pastelred, scale=2] {$\frac{n_1n_2}{r^2}$}(10,-0.1);
	
	\node[align=left,below,scale=1.25] at (2.5, -0.5) {$LB_{GEMM}=n_2r-\frac{n_2r}{P} \quad$};
	\node[align=left,below,scale=1.25] at (7.5, -0.8) {$LB_{GEMM}=2\left(\frac{n_1n_2r^2}{P}\right)^{1/2}$\\ $\qquad\qquad\qquad-\frac{n_1r+n_2r}{P}$};
	\node[align=center,below,scale=1.25] at (13.25, -0.8) {$LB_{GEMM}= 3\left(\frac{n_1n_2r}{P}\right)^{2/3}\qquad$ \\$\qquad\qquad\quad-\frac{n_1r+n_2r+n_1n_2}{P}$};	
	\end{tikzpicture}
\end{center}
For $n_2 > n_1 \geq r$, the same visualization applies upon interchanging $n_1$ and $n_2$.

\subsection{Optimal Parallel Algorithm}
\label{sec:randMatmulAlg}
In this section, we present an optimal algorithm to show that the lower bound of \cref{thm:randMatmulLB} is tight. 
We organize $P$ processors into a $p_1\times p_2\times p_3$ processor grid and assign the computation to processors according to their coordinates, represented as $(i,j,k)$ which give their locations in the grid. 
The algorithm performs two collective operations, one All-Gather and one Reduce-Scatter. 
Each processor receives the portion of the input matrix $\M{A}$ it needs to perform its computation through an All-Gather collective operation. 
After that, the processor generates the required portion of the random matrix and performs its local computation. 
The result of each local computation must be summed with all other contributions to the same output matrix entries from other processors, which is achieved by a Reduce-Scatter collective operation.

As we require that there is one copy of data in the system at the start and end of the computation, we distribute the input (resp. output) matrix $\M{A}$ (resp. $\M{B}$) evenly among all processors at the start (resp. end). 
For simplicity of explanation, we assume that $p_1, p_2, p_3$ evenly divide $n_1, n_2, r$, respectively. 
We use the notation $\M{A}_{ij}$ to denote the submatrix of $\M{A}$ such that 
$$\M{A}_{ij} = \M{A}\left((i-1)\cdot \frac{n_1}{p_1}+1:i\cdot \frac{n_1}{p_1}, (j-1)\cdot \frac{n_2}{p_2}+1:j\cdot \frac{n_2}{p_2}\right)\text.$$
We assume that $\M{A}_{ij}$ is distributed evenly among $\Pi_{ij*}$ processors at the beginning of the computation and denote the part of $\M{A}_{ij}$ owned by the processor $(i,j,k)$ by $\M{A}_{ij}^{(k)}$. 
We define $\M{B}_{ik}$ similarly to $\M{A}_{ij}$. 
At the end, $\M{B}_{ik}$ is distributed evenly among  $\Pi_{i*k}$ processors and we denote the portion of $\M{B}_{ik}$ owned by processor $(i,j,k)$ with $\M{B}_{ik}^{(j)}$.

\begin{algorithm}[tb]
	\caption{\label{alg:3dRandMatmul}Matrix Multiplication with a Random Matrix}
	\begin{algorithmic}[1]		
		\Require $\Pi$ is a $p_1\times p_2 \times p_3$  grid of processors, $|\Pi| = P$.
		\Require $\M{A}$ is evenly divided into a $p_1\times p_2$ grid of rectangular blocks of dimension $n_1/p_1\times n_2/p_2$, and each block $\M{A}_{ij}$ is evenly divided across a set of $p_3$ processors. $\M{A}_{ij}^{(k)}$ is owned by processor coordinate $(i,j,k)$.
		\Ensure $\M{B} = \M{A}\M{\Omega}$, $\M{B}$ is evenly divided across a $p_1\times p_3$ grid of blocks, and each block $\M{B}_{ik}$ is evenly divided across a set of $p_2$ processors. $\M{B}_{ik}^{(j)}$ is owned by processor coordinate $(i,j,k)$.
		\Function{$\M{B}_{ik}^{(j)}=$ RandMatMul}{$\M{A}_{ij}^{(k)},\Pi$}
		\State $(i,j,k) = \Call{MyCoordinate}{\Pi}$
		\State //Gather the required data of input matrix $\M{A}$
		\State $\M{A}_{ij}$ = \Call{All-Gather}{$\M{A}_{ij}^{(k)}$, $\Pi_{ij*}$}\label{alg:3dmatmul:line:allGatherMatrixA}
		\State //Generate the required random submatrix
		\State $\M{\Omega}_{jk}$ =  \Call{GenRandom}{$n_2/p_2,r/p_3$}
		\State //Perform local matrix multiplication
		\State $\bar{\M{B}_{ik}}=\M{A}_{ij} \cdot \M{\Omega}_{jk}$\label{alg:3dmatmul:line:localComputation}
		\State //Sum results to compute $\M{B}_{ik}^{(j)}$
		\State $\M{B}_{ik}^{(j)} = \Call{Reduce-Scatter}{\bar{\M{B}_{ik}},\Pi_{i*k}}$\label{alg:3dmatmul:line:ReduceScatterMatrixB}
		\EndFunction
	\end{algorithmic}
\end{algorithm}

\subsubsection{Cost Analysis}
\label{sec:randMatmulAlg:cost}
We now analyze computation and communication costs of \cref{alg:3dRandMatmul}. 
Each processor performs $(n_1/p_1)(n_2/p_2)(r/p_3)\allowbreak =n_1n_2r/P$ scalar multiplications in \cref{alg:3dmatmul:line:localComputation}. 

Communication occurs only in the All-Gather and Reduce-Scatter collectives in \cref{alg:3dmatmul:line:allGatherMatrixA,alg:3dmatmul:line:ReduceScatterMatrixB}, respectively.
The bandwidth costs of \cref{alg:3dmatmul:line:allGatherMatrixA,alg:3dmatmul:line:ReduceScatterMatrixB} in \cref{alg:3dRandMatmul} are $\left(1-\frac{1}{p_3}\right)\frac{n_1n_2}{p_1p_2}$ and $\left(1-\frac{1}{p_2}\right)\frac{n_1r}{p_1p_3}$, respectively. 
Thus the overall bandwidth cost of \cref{alg:3dRandMatmul} is $\left(1-\frac{1}{p_3}\right)\frac{n_1n_2}{p_1p_2} + \left(1-\frac{1}{p_2}\right)\frac{n_1r}{p_1p_3} = \frac{n_1n_2}{p_1p_2} + \frac{n_1r}{p_1p_3} - \frac{n_1n_2+n_1r}{P}.$ 
The latency costs of \cref{alg:3dmatmul:line:allGatherMatrixA,alg:3dmatmul:line:ReduceScatterMatrixB} are $\log(p_3)$ and $\log(p_2)$, respectively. 
Thus the overall latency cost of the algorithm is $\log(p_3) + \log(p_2)$.

\subsection{Optimal Processor Grid Selection}
\label{sec:randMatmulAlg:optprocgrid}
To minimize the communication costs of our algorithm, we select $p_1, p_2$ and $p_3$ based on the terms of the communication lower bound (\cref{thm:randMatmulLB}). 
As the lower bound has three cases, we discuss for each separately.
In all cases, we assume that the numerator is divisible by the denominator for each division expression. 

\paragraph{Case 1 ($P\leq n_1$)}
We set $p_1=P,$ and $p_2= p_3 = 1$. 
Thus the bandwidth cost is $(n_1n_2+n_1r)/P-(n_1n_2+n_1r)/P = 0$, so no communication is required matching the lower bound.

\paragraph{Case 2 ($n_1 < P \leq n_1n_2/r$)}
We set $p_1 = n_1,$ $p_2 = P/n_1,$ and $p_3 = 1.$ 
The bandwidth cost is $n_1n_2/P + r - (n_1n_2+n_1r)/P$, which matches the lower bound.

\paragraph{Case 3 ($n_1n_2/r <  P$)}
We set $p_1 = n_1,$  $p_2 = (Pn_2/(rn_1))^{1/2},$ and $p_3 = (Pr/(n_1n_2))^{1/2}$.
The bandwidth cost is $2(n_1n_2r/P)^{1/2}-(n_1n_2+n_1r)/P$, which matches the lower bound.

\section{Nystr\"{o}m Approximation}
\label{sec:nystrom}
In this section we consider an algorithm that performs two multiplications in sequence which are key to computing the Nystr\"{o}m approximation, $\M{B} = \M{A}\M{\Omega}$ and $\M{C} = \M{\Omega}^T\M{B}$.
Here $\M{A}$ is a symmetric matrix of dimensions $n\times n$ and $\M{\Omega}$ is a random matrix of dimensions $n\times r$ with $r<n$.  
The first multiplication requires $n^2r$ scalar multiplications, while the second requires $nr^2$ scalar multiplications.
Our method of proving the bound is able to handle the sequential nature of the computation through the assumption that each of the component computations is load balanced.
Our bound improves over the sum of the best lower bound for each component computation by taking into account that the data a processor requires from $\M{B}$ must be the union of the data required for each of the component computations.

\subsection{Communication Lower Bound}
\label{sec:nystrom:lb}
As with our previous bound, we begin by presenting the abstract optimization problem which allows us to determine the minimum amount of data a processor will require.

Similar to the proof of \cref{lem:randMatmulOpt}, the following lemma can be established by appropriately choosing the dual variables $\vc{\mu^*}$ and verifying that all the KKT conditions are satisfied.

\begin{lemma}
	\label{lem:NystromOpt}
	Consider the following optimization problem
	$$\min_{\V{x}\in\mathbb{R}^3} x_1 + x_2 + x_3 $$
	such that 
	$x_1x_2 \geq n^2r/P, x_2x_3 \geq nr^2/P, n^2/P \leq x_1, nr/P \leq x_2, \allowbreak r^2/P\leq x_3,$
	where $P,n$, and $r$ are all positive integers greater than or equal to 1, and $r<n$.
	The constraints induce four cases for the optimal solution $\V{x}^*$ based upon the relative sizes of $n$ and $r$:
	\begin{itemize}
		\item If $P\leq r$ then $\V{x}^*=\begin{bmatrix} n^2/P & nr/P & r^2/P\end{bmatrix};$
		\item If $r < P \leq n$ then $\V{x}^* =\begin{bmatrix} n^2/P & nr/P & r \end{bmatrix};$
		\item If $n < P \leq n(n+r)/r$ then $\V{x}^* = \begin{bmatrix} n^2/P & r & nr/P\end{bmatrix};$
		\item If $n(n+r)/r < P $ then $$\V{x}^* = \begin{bmatrix} n\left(\frac{nr}{(n+r)P}\right)^{1/2} & \left(\frac{nr(n+r)}{P}\right)^{1/2} & r\left(\frac{nr}{(n+r)P}\right)^{1/2}\end{bmatrix}.$$
	\end{itemize}
	This can be visualized as follows:
	
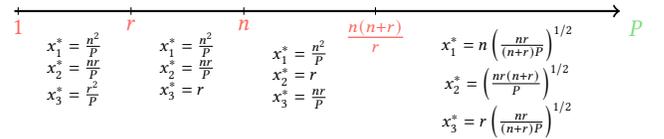
\begin{center}
		\begin{tikzpicture}[scale=0.5, every node/.style={transform shape}]
		\draw [->, thick] (-0.1,0) -- (15,0) node [above,scale=2] {Increasing $P$};
		
		\draw (0, 0.1) -- node [below, pastelred, scale=2]{$1$}(0,-0.1);
		\draw (3, 0.1) -- node [below, pastelred, scale=2]{$r$}(3,-0.1);
		\draw (6, 0.1) -- node [below, pastelred, scale=2] {$n$}(6,-0.1);
		\draw (9.5, 0.1) -- node [below, pastelred, scale=2] {$\frac{n(n+r)}{r}$}(9.5,-0.1);
		
		\node[align=left,below,scale=1.5] at (1.5, -0.4) {$x_1^*=\frac{n^2}{P}$\\ $x_2^*=\frac{nr}{P}$\\ $x_3^*=\frac{r^2}{P}$};
		\node[align=left,below,scale=1.5] at (4.5, -0.4) {$x_1^*=\frac{n^2}{P}$\\ $x_2^*=\frac{nr}{P}$\\ $x_3^*=r$};		\node[align=left,below,scale=1.5] at (7.5, -0.6) {$x_1^*=\frac{n^2}{P}$\\ $x_2^*=r$\\ $x_3^*=\frac{nr}{P}$};
		\node[align=center,below,scale=1.5] at (13, -0.2) {$x_1^*=n\left(\frac{nr}{(n+r)P}\right)^{1/2}$\\ $x_2^*=\left(\frac{nr(n+r)}{P}\right)^{1/2}$\\ $x_3^*= r\left(\frac{nr}{(n+r)P}\right)^{1/2}$};	
		\end{tikzpicture}
	\end{center}
\end{lemma}

\begin{proof}
	To begin note that the first and second constraints are  quasiconvex in the positive quadrant by \cref{lem:quasiconvex:xi}, and the objective function and all remaining constraints are affine functions. 
	Thus we can establish the optimality of the solution for every case by showing that the KKT conditions specified in \cref{def:KKT} are satisfied by \cref{lem:KKT}.

	To match standard notation, let $f(\V{x}) = x_1+x_2+x_3$ and
	$$g(\V{x}) = \begin{bmatrix}
	\frac{n^2r}{P}-x_1x_2 \\
	\frac{nr^2}{P}-x_2x_3 \\
	\frac{n^2}{P}-x_1 \\
	\frac{nr}{P}-x_2 \\
	\frac{r^2}{P}-x_3
	\end{bmatrix}.$$
	Then $\nabla f(\V{x}) = \begin{bmatrix}1 & 1 & 1 \end{bmatrix}$ and
	$$J_g(\V{x}) = \begin{bmatrix}
	-x_2 & -x_1 & 0 \\
	0 & -x_3 & -x_2 \\
	-1 & 0 & 0 \\
	0 & -1 & 0 \\
	0 & 0 & -1 \\
	\end{bmatrix}.
	$$
	For each case, we will give the optimal primal and dual solutions, and comment on why primal and dual feasibility conditions are satisfied.
	In every case stationarity can be directly verified, and complementary slackness holds as only dual variables $\mu_i$ corresponding to tight constraints are not zero.
	\paragraph{Case 1 ($P\leq r$)}
	Set $\V{x}^* = \begin{bmatrix} n^2/P & nr/P & r^2/P \end{bmatrix}$ and $\vc{\mu}^* = \linebreak[4] \begin{bmatrix}  0 & 0 & 1 & 1 & 1 \end{bmatrix}.$
	Primal feasibility holds because $P\leq r \leq n$ implies that $n^2r/P \leq n^2/P \cdot nr/P $ and $nr^2/P \leq nr/P \cdot r^2/P\text.$ Dual feasibility is immediate.
	\paragraph{Case 2 ($r < P \leq n$)}
	Set $\V{x}^* = \begin{bmatrix} n^2/P & nr/P & r \end{bmatrix}$ and $\vc{\mu}^* = \begin{bmatrix} 0 & P/(nr) & 1 & 1-P/n & 0 \end{bmatrix}.$
	Primal feasibility holds because $P \leq n$ implies that $n^2r/P \leq n^2/P \cdot nr/P$ and $r < P $ implies that $r^2/P < r$.
	Dual feasibility holds because $P \leq n.$
	\paragraph{Case 3 ($n < P \leq n(n+r)/r$)}
	Set $\V{x}^* = \begin{bmatrix} n^2/P & r & nr/P \end{bmatrix}$ and $\vc{\mu}^* = \begin{bmatrix}(P-n)/n^2  & 1/r & (n(n+r)/r -P)r/n^2 & 0 & 0\end{bmatrix}.$
	Primal feasibility holds because $n < P$ implies $nr/P < r$ and $r\leq n$ implies $r^2/P \leq nr/P$.
	Dual feasibility holds because $n < P$ and $P \leq n(n+r)/r$.
	\paragraph{Case 4 ($n(n+r)/r< P$)}
	Set $\V{x}^* = \begin{bmatrix}nt & (n+r)t & rt\end{bmatrix}$ \linebreak[4]
	and $\vc{\mu}^* = \begin{bmatrix} 1/((n+r)t) & 1/((n+r)t) & 0 & 0 & 0 \end{bmatrix} $ where  $t= \linebreak[4] (nr/((n+r)P))^{1/2}$.
	Primal feasibility holds because $r\leq n$ and $n(n+r)/r < P$ imply that $n^2/P < nt$, $nr/P < (n+r)t$ and $r^2/P < rt$. Dual feasibility is immediate.
\end{proof}

We prove the following theorem using the solution to the optimization problem in \cref{lem:NystromOpt}.
A visualization of the iteration space for the two multiplications and the related operands is given in \cref{fig:Nystrom:iterspace}, along with two valid parallelizations across 3 processors.

\begin{theorem}
	\label{thm:NystromLB}
	Consider the sequence of two computations, $\M{B}=\M{A}\cdot\M{\Omega}$ and $\M{C}=\M{\Omega^\Tra}\cdot\M{B}$ using classical matrix multiplication, where $\M{A}$, $\M{B}$ and $\M{C}$ have dimensions $n\times n$, $n\times r$ and $r \times r$, respectively, and $\M{\Omega}$ is a random matrix of dimensions $n \times r$ with $n>r$. 
	Any parallel algorithm using $P$ processors that load balances each multiplication and begins with one copy of the input matrix $\M{A}$ and ends with one copy of the output matrices $\M{B}$ and $\M{C}$ must communicate at least $ W-(n^2+nr+r^2)/P$ words of data where
	$$W = \begin{cases}
	\quad \frac{n^2+nr+r^2}{P} & \text{ if } \quad 1\leq P \leq r\\
	\quad \frac{n^2+nr}{P}+r & \text{ if } \quad r < P \leq n\\
	\quad \frac{n^2}{P}+r+\frac{nr}{P} & \text{ if } \quad n < P \leq n(n+r)/r\\
	\quad 2\left(\frac{nr(n+r)}{P}\right)^{1/2} & \text{ if } \quad n(n+r)/r < P\text.\\
	\end{cases}$$
\end{theorem}

\begin{proof}
As the algorithm load balances each multiplication, each processor must perform $n^2r/P$ scalar multiplications corresponding to iteration points of the first multiplication, and $nr^2/P$ scalar multiplications corresponding to iteration points of the second multiplication. 
Because the algorithm begins and ends the computation with at most one copy of the input and output matrices, there must be at least one processor that owns at most $(n^2+nr+r^2)/P$ elements of the input and output matrices.
We consider this processor.

Let $F$ be the set of iteration points the processor performs from the first multiplication, and $F'$ the set of iteration points the processor performs from the second multiplication.
So $|F| = n^2r/P,$ and $|F'| = nr^2/P.$
We let $\phi_{ij}(F)$ and  $\phi_{ik}(F)$ denote the projections of $F$ onto the input matrix $\M{A}$ and the output matrix $\M{B}$, and $\phi_{i'k'}(F')$ and $\phi_{j'k'}(F')$ denote the projections of $F'$ onto $\M{B}$ and $\M{C}$. 
Then \cref{thm:hbl:random} implies that  $|\phi_{ij}(F)||\phi_{ik}(F)|\geq |F| = n^2r/P$ and $ |\phi_{j'k'}(F')||\phi_{i'k'}(F')| \geq nr^2/P$. 
By \cref{lem:proj} we know that $|\phi_{ik}(F)| \geq nr/P$ and $|\phi_{i'k'}(F')| \geq nr/P$.
As $|\phi_{ik}(F)\cup\phi_{i'k'}(F')| \geq \max(|\phi_{ik}(F)|, |\phi_{i'k'}(F')|)$, we have $|\phi_{ij}(F)||\phi_{ik}(F)\cup\phi_{i'k'}(F')| \geq n^2r/P$, $|\phi_{j'k'}(F')||\phi_{ik}(F)\cup\phi_{i'k'}(F')|\geq nr^2/P$ and $|\phi_{ik}(F)\cup\phi_{i'k'}(F')| \allowbreak \geq nr/P$. 
By \cref{lem:proj}, we also know that $|\phi_{ij}(F)|\geq n^2/P,$ and $|\phi_{j'k'}(F')|\geq r^2/P$.

Thus, to minimize the communication we need to minimize the number of elements accessed by this processor, which corresponds to minimizing $|\phi_{ij}(F)| + |\phi_{ik}(F)\cup\phi_{i'k'}(F')| + |\phi_{j'k'}(F')|$ subject to the above constraints. 
Hence the solution to \cref{lem:NystromOpt} gives the minimum number of elements that must be accessed by this processor, and subtracting the data the processor can own from the result proves the lower bound. 
\end{proof}

\subsection{Parallel Algorithm}
\label{sec:nystrom:algo}
We organize $P$ processors into a $p_1\times p_2\times p_3$ processor grid to perform the first computation ($\M{B}=\M{A}\M{\Omega}$) and into a $q_1\times q_2\times q_3$ processor grid to perform the second computation ($\M{C}=\M{\Omega^\Tra}\M{B}$). 
The coordinate of a processor is represented as $(i,j,k)$ in the first grid and $(i',j',k')$ in the second grid. 
Both computations are performed in the same way, but their grid layouts may be different. We present our algorithm in \cref{alg:nystrom}.

\begin{algorithm}[htb]
	\caption{\label{alg:nystrom}Parallel algorithms for $\M{B}=\M{A}\M{\Omega}$ and $\M{C}=\M{\Omega^\Tra}\M{B}$.}
	\begin{algorithmic}[1]		
		\Require $\Pi$ is a $p_1\times p_2 \times p_3$  grid of processors and $\Psi$ is a $q_1 \times q_2 \times q_3$ grid of processors, and $|\Pi| = |\Psi| = P$.
		\Require $\M{A}$ is evenly divided into a $p_1\times p_2$ grid of rectangular blocks of dimension $n/p_1\times n/p_2$, and each block $\M{A}_{ij}$ is evenly divided across a set of $p_3$ processors with $\M{A}_{ij}^{(k)}$ owned by processor coordinate $(i,j,k)$ in $\Pi$.
		\Ensure $\M{B} = \M{A}\M{\Omega}$, $\M{B}$ is evenly divided across a $q_1\times q_3$ grid of blocks, and each block $\M{B}_{i'k'}$ is evenly divided across a set of $q_2$ processors with $\M{B}_{i'k'}^{(j')}$  owned by processor coordinate $(i',j',k')$ in $\Psi$.
		\Ensure $\M{C} = \M{\Omega^\Tra}\M{B}$, $\M{C}$ is evenly divided across a $q_2\times q_3$ grid of blocks, and each block $\M{C}_{j'k'}$ is evenly divided across a set of $q_1$ processors with $\M{C}_{j'k'}^{(i')}$ owned by processor coordinate $(i',j',k')$ in $\Psi$.
		\Function{$(\M{B}_{i'k'}^{(j')},\M{C}_{j'k'}^{(i')})=$ RandCompNys}{$\M{A}_{ij}^{(k)},\Pi, \Psi$}
		\State $(i,j,k) = \Call{MyCoordinate}{\Pi}$
		\State //Gather the required data of input matrix $\M{A}$
		\State $\M{A}_{ij}$ = \Call{All-Gather}{$\M{A}_{ij}^{(k)}$, $\Pi_{ij*}$}\label{alg:nystrom:line:allGatherMatrixA}
		\State //Generate the required $n/p_2 \times r/p_3$ portion of random matrix for $\M{B}=\M{A}\M{\Omega}$
		\State $\M{\Omega}_{jk}$ =  \Call{GenRandom}{$n/p_2,r/p_3$}
		\State //Perform first local matrix multiplication
		\State $\M{\bar{\hat{B}}}_{ik}=\M{A}_{ij} \cdot \M{\Omega}_{jk}$\label{alg:nystrom:line:firstlocalComputation}
		\State //Each block $\M{\hat{B}}_{ik}$ is evenly divided across a set of $p_2$ processors with $\M{\hat{B}}_{ik}^{(j)}$ is owned by processor coordinate $(i,j,k)$
		\State //Sum results to compute $\M{\hat{B}}_{ik}^{(j)}$
		\State $\M{\hat{B}}_{ik}^{(j)} = \Call{Reduce-Scatter}{\M{\hat{B}}_{ik},\Pi_{i*k}}$\label{alg:nystrom:line:ReduceScatterMatrixB}
		\State $(i',j',k') = \Call{MyCoordinate}{\Psi}$
		\State //Change distribution of $\M{\hat{B}}$ so that it is suitable for $q_1 \times q_2 \times q_3$ grid 
		\If{\emph{for any} $i, p_i \neq q_i$}
		\State $\M{B}_{i'k'}^{(j')} $ = \Call{Redistribute}{$\M{\hat{B}}$} \label{alg:nystrom:line:RedistMatrixB}
		\Else
		\State $\M{B}_{i'k'}^{(j')} = \M{\hat{B}}_{ik}^{(j)}$
		\EndIf 
		\State //Generate the required $n/q_1 \times r/q_2$ portion of random matrix for $\M{C}=\M{\Omega^\Tra}\M{B}$
		\State $\M{\Omega}_{i'j'}$ =  \Call{GenRandom}{$n/q_1,r/q_2$}
		\State //Gather the required data of matrix $\M{B}$
		\State $\M{B}_{i'k'}$ = \Call{All-Gather}{$\M{B}_{i'k'}^{(j')}$, $\Psi_{i'*k'}$}\label{alg:nystrom:line:allGatherMatrixB}
		\State //Perform second local matrix multiplication
		\State $\M{\bar{C}}_{j'k'}=\M{\Omega}_{i'j'}^\Tra \cdot \M{B}_{i'k'}$\label{alg:nystrom:line:secondlocalComputation}
		\State //Sum results to compute $\M{C}_{j'k'}^{(i')}$
		\State $\M{C}_{j'k'}^{(i')} = \Call{Reduce-Scatter}{\M{\bar{C}}_{j'k'},\Psi_{*j'k'}}$\label{alg:nystrom:line:ReduceScatterMatrixC}
		\EndFunction
	\end{algorithmic}
\end{algorithm}

In the $p_1\times p_2\times p_3$ grid, the algorithm performs two collective operations, All-Gather and Reduce-Scatter. 
Each processor receives the portion of the input matrix $\M{A}$ it needs to perform all of its computation through an All-Gather collective operation. 
After that, the processor generates the required portion of the random matrix and performs its local computation. 
The result of each local computation must be summed with all other contributions to the same output matrix entries from other processors to obtain $\M{B}$, and it is achieved by a Reduce-Scatter collective operation. 
If layout of both processor grids are different, then $\M{B}$ is redistributed according to the $q_1\times q_2\times q_3$ grid such that the second computation can be performed in this grid.

In the $q_1 \times q_2 \times q_3$ grid, the algorithm also performs an All-Gather and a Reduce-Scatter collective operations. 
Each processor receives the required portion of matrix $\M{B}$ through an All-Gather collective operation. 
After that, the processor generates the required portion of the random matrix and performs its local computation. 
In the end, a Reduce-Scatter collective operation is performed to obtain the final matrix $\M{C}$.

We use the same data distribution and notation as in~\cref{sec:randMatmulAlg}, and again impose that there is one copy of data in the system. 
The input matrix $\M{A}$ is distributed evenly among all processors at the start, and the output matrices $\M{B}$ and $\M{C}$ are distributed evenly at the end. 
We assume that $\M{A}_{ij}$ is distributed evenly among $\Pi_{ij*}$ processors at the beginning of the computation and denote the portion of $\M{A}_{ij}$ owned by processor $(i,j,k)$ by $\M{A}_{ij}^{(k)}$. 
At the end of the computation, $\M{B}_{i'k'}$ and $\M{C}_{j'k'}$ are distributed evenly among  $\Psi_{i'*k'}$ and $\Psi_{*j'k;}$ processors. 
We denote the portion of $\M{B}_{i'k'}, \M{C}_{j'k'}$ owned by processor $(i',j',k')$ by $\M{B}_{i'k'}^{(j')}$ and $\M{C}_{j'k'}^{(i')}$.

\subsubsection{Cost Analysis}
\label{sec:nystrom:algo:cost}

We now analyze computation and communication costs of the algorithm. 
Each processor performs $\frac{n}{p_1}\cdot\frac{n}{p_2}\cdot\frac{r}{p_3}$ and $\frac{r}{q_2}\cdot\frac{n}{q_1}\cdot\frac{r}{q_3}$ scalar multiplications in \cref{alg:nystrom:line:firstlocalComputation,alg:nystrom:line:secondlocalComputation}, respectively. 
Thus the total number of multiplications performed by a processor is $\frac{n^2r}{p_1p_2p_3} + \frac{nr^2}{q_1q_2q_3} = \frac{nr(n+r)}{P}$.
 
Communication occurs in the All-Gather collective in \cref{alg:nystrom:line:allGatherMatrixA,alg:nystrom:line:allGatherMatrixB} and the Reduce-Scatter collective in \cref{alg:nystrom:line:ReduceScatterMatrixB,alg:nystrom:line:ReduceScatterMatrixC}. 
In addition, if the two grid layouts differ then communication is also required to redistribute matrix $\M{B}$ in \cref{alg:nystrom:line:RedistMatrixB}.

The bandwidth costs of \cref{alg:nystrom:line:allGatherMatrixA,alg:nystrom:line:allGatherMatrixB,alg:nystrom:line:ReduceScatterMatrixB,alg:nystrom:line:ReduceScatterMatrixC} in \cref{alg:nystrom} are $\left(1-\frac{1}{p_3}\right)\frac{n^2}{p_1p_2}$, $\left(1-\frac{1}{p_2}\right)\frac{nr}{p_1p_3}$, $\left(1-\frac{1}{q_2}\right)\frac{nr}{q_1q_3}$ and $\left(1-\frac{1}{q_1}\right)\frac{r^2}{q_2q_3}$, respectively. 
Thus the overall bandwidth cost of \cref{alg:nystrom} is 
$\left(1-\frac{1}{p_3}\right)\frac{n^2}{p_1p_2} + \left(1-\frac{1}{p_2}\right)\frac{nr}{p_1p_3} + \left(1-\frac{1}{q_2}\right)\frac{nr}{q_1q_3} + \left(1-\frac{1}{q_1}\right)\frac{r^2}{q_2q_3} + BW(Redistr)$.

Here $BW(Redistr)$ denotes the bandwidth cost to redistribute matrix $\M{B}$ according to the second grid layout. 
Note that both $\M{\hat{B}}_{ik}^{(j)}$ and $\M{B}_{i'k'}^{(j')}$ in \cref{alg:nystrom} contain $nr/P$ elements. 
Thus each processor needs to send at max $nr/P$ words and receive at max $nr/P$ words. 
Thus the bandwidth cost to redistribute matrix $\M{B}$ in a fully connected network is at max $nr/P$~\cite{Thakur:CollectiveCommunications:2005,Chan:CollectiveCommunications:2007}.
 
The latency costs of \cref{alg:nystrom:line:allGatherMatrixA,alg:nystrom:line:allGatherMatrixB,alg:nystrom:line:ReduceScatterMatrixB,alg:nystrom:line:ReduceScatterMatrixC} in \cref{alg:nystrom} are $\log(p_3)$, $\log(q_2)$, $\log(p_2)$ and $\log(q_1)$, respectively. 
Thus the overall latency cost of the algorithm is $\log(p_3) + \log(p_2) + \log(q_2) + \log(q_1) + L(Reddistr)$. Here $L(Redistr)$ denotes the latency cost to redistribute matrix $\M{B}$ according to the second grid layout. In a fully connected network of $P$ processors using a point-to-point algorithm, it is $O(P)$.
The latency cost can be reduced to $O(\log P)$ using a bidirectional-exchange algorithm at the expense of a bandwidth cost that grows by a factor of $O(\log P)$ \cite{Thakur:CollectiveCommunications:2005,Chan:CollectiveCommunications:2007}.

\subsection{Processor Grid Selection}
We consider two approaches to selecting the processor grid dimensions.
In the first, we select the grid dimensions $p_i$ and $q_i$ based on the lower bounds (\cref{thm:NystromLB}), whereas in the second we select the dimensions so that $\M{B}$ does not need to be redistributed.

In our first approach of selecting a processor grid, we do not expect the bandwidth cost of our algorithm to be optimal because the lower bounds have one term to express the access of matrix $\M{B}$ but in our algorithm communication happens twice for $\M{B}$. 
Additionally, our lower bound frameworks do not model redistribution cost which we need to perform.
The lower bound has four cases, we discuss for each one separately.
  
\paragraph{Case 1 ($P\leq r$)}
We set $p_1=P, p_2=p_3=1$ and $q_1= q_2 = 1, q_3=P$. 
Thus the bandwidth cost is $n^2/P+2nr/P+r^2/P-(n^2+2nr+r^2)/P + BW(Redistr)  \leq nr/P$, which is at most $nr/P$ greater than the lower bound.

\paragraph{Case 2 ($r < P \leq  n$)}
We set $p_1=P, p_2=p_3=1$ and $q_1=P/r, q_2 = 1, q_3=r$. 
The bandwidth cost is $n^2/P+2nr/P+ r -(n^2+2nr+r^2)/P + BW(Redistr)  \leq r-r^2/P+ nr/P$, which is at most $nr/P$ greater than the lower bound.

\paragraph{Case 3 ($n < P \leq  n(n+r)/r$)}
We set $p_1=n, p_2=P/n,p_3=1$ and $q_1=n/r, q_2 = P/n, q_3=r$. 
The bandwidth cost is $n^2/P+2r+nr/P -(n^2+2nr+r^2)/P + BW(Redistr)  \leq 2r-r^2/P$, which is at most $r$ greater than the lower bound.

\paragraph{Case 4 ($n(n+r/r) <  P$)}
We set $p_1=n, p_2=((n+r)P/(nr))^{1/2},p_3=(rP/(n(n+r)))^{1/2}$ and $q_1=(nP/(r(n+r)))^{1/2} n/r, q_2 = ((n+r)P/(nr))^{1/2}, q_3=r$. 
The bandwidth cost is $n(nr/((n+r)P))^{1/2} + 2(nr(n+r)/P)^{1/2} + r(nr/((n+r)P))^{1/2} -(n^2+2nr+r^2)/P + BW(Redistr)  \leq 3(nr(n+r)/P)^{1/2} - (n^2+nr+r^2)/P $, which is at most $(nr(n+r)/P)^{1/2}$ greater than the lower bound.

The second approach is based on the fact that $\M{B} = \M{A}\M{\Omega}$ is the most expensive part. We set $p_i$ based on~\cref{sec:randMatmulAlg:optprocgrid}, and take $q_i=p_i$. Note that this approach does not require redistribution for matrix $\M{B}$. 
For $P\leq n$, we set $p_1=q_1=P$ and $p_2=q_2=p_3=q_3=1$. 
The bandwidth cost of our algorithm is $n^2/P + nr/P + r^2 - (n^2+nr+r^2)/P = r^2 - r^2/P$. 
In certain settings, particularly when $P$ is smaller than $n/r$, this cost is lower than that of Case 1 of the first approach.

We implement two variants of the algorithm based on 1D matrix multiplications and compare their performance empirically in \cref{sec:experiments}.\linebreak
The first variant is based on Case 1 of the first approach ($p_1=q_3=P$) and we call it \emph{Redist} algorithm.
The second variant is based on the second approach with $p_1=q_1=P$ and we call it \linebreak\emph{No-Redist} algorithm.
These are the two most efficient variants when $P<r$, which we expect to hold in most practical situations.
\Cref{fig:Nystrom:iterspace} shows distribution of the computation for both variants across 3 processors.

\begin{figure}[tb]
	\centering
	\input{fig/nystromViz.tex}
	\caption{Iteration space of $\M{B} = \M{A}\M{\Omega}$ and $\M{\Omega}^\Tra\M{B}=\M{C}$ computation with a total of $n(n+r)r$ iteration points. The faces show the accesses to different matrices and the shading corresponds to distribution of the computation across 3 processors.  The prism on the left depicts the algorithm with $p_1=q_3=P$ (1D algorithms with \emph{Redist}ribution of $\M{B}$), and the prism on the right depicts the algorithm with $p_1=q_1=P$ (1D algorithms with \emph{No-Redist}ribution of $\M{B}$). \label{fig:Nystrom:iterspace}}
\end{figure}
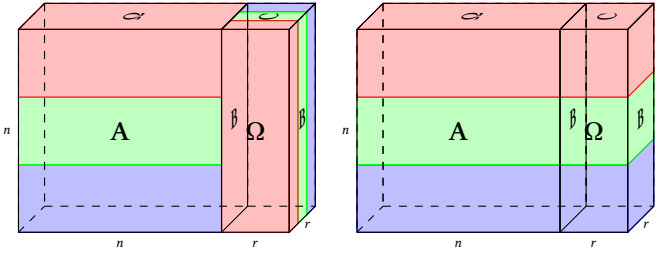

\section{Experiments}
\label{sec:experiments}
\input{experiments.tex}

\section{Conclusion}
\label{sec:conclusion}

In this work, we establish communication lower bounds for the multiplication of a matrix with a random matrix and the computations of Nystr\"{o}m approximation. Our lower bound proofs rely on a new geometric inequality that relates the size of a $3$-dimensional set to its $2$-dimensional projections. We employ this inequality, along with additional constraints, to formulate communication lower bounds as constrained optimization problems. We solve these problems analytically to derive bounds. For the multiplication of a matrix with a random matrix, we show that our bounds are tight in all ranges by presenting a communication optimal algorithm. 
We also demonstrate that the communication cost of our algorithm for Nystr\"{o}m approximation is close to the lower bound. 

We demonstrate that no communication is necessary for computing $\M{A}\M{\Omega}$ using a 1D algorithm when the number of processors does not exceed the number of rows of $\M{A}$.
In the case of Nystr\"{o}m approximation, we implement two algorithms based on 1D matrix multiplication that are the most efficient variants in practical situations. 
Our numerical experiments on modern state-of-the-art supercomputing systems equipped with CPUs and GPUs demonstrate their parallel scalability.

We are interested in extending our approaches to other sketching methods that use structured matrices such as count sketch, more general sparse embeddings, and subsampled randomized Hadamard transform methods. 
In Nystr\"{o}m approximation, both input  and output matrices are symmetric. 
We plan to exploit symmetry to further reduce communication and computation costs in the future. 
In this work, we assume that processors load balance each multiplication of the Nystr\"{o}m approximation. 
We are also interested to study the case where the overall computation is balanced.

\begin{acks}
This project has received funding through the UKRI Digital Research Infrastructure Programme through the Science and Technology Facilities Council’s Computational Science Centre for Research Communities (CoSeC).
This project received funding from the European Research Council (ERC) under European Union's Horizon 2020 research and innovation program (grant agreement 810367) and the Platform for Advanced Scientific Computing (PASC) project RandESC.
This work is supported by the National Science Foundation under grant CCF-1942892. This material is based upon work supported by the US Department of Energy, Office of Science, Advanced Scientific Computing Research program under awards DE-SC0023296 and DE-SC0025394.
\end{acks}
%


\bibliographystyle{ACM-Reference-Format}
\bibliography{refs}

\end{document}

%% file: abstract.tex
Sketching is widely used in randomized linear algebra for low-rank matrix approximation, column subset selection, and many other problems, and it has gained significant traction in machine learning applications.
However, sketching large matrices often necessitates distributed memory algorithms, where communication overhead becomes a critical bottleneck on modern supercomputing clusters. 
Despite its growing relevance, distributed-memory parallel strategies for sketching remain largely unexplored.
In this work, we establish communication lower bounds for sketching using dense matrices that determine how much data movement is required to perform it in parallel. 
One important observation of our lower bounds is that no communication is required for a small number of processors. 
We show that our lower bounds are tight by presenting communication optimal algorithms. 
Furthermore, we extend our approach to determine communication lower bounds for computations of Nyström approximation where sketching is applied twice. 
We also introduce novel parallel algorithms whose communication costs are close to the lower bounds. 
Finally, we implement our algorithms on modern state-of-the-art supercomputing infrastructures which have both CPU- and GPU-equipped systems and demonstrate their parallel scalability.

%% file: fig/nystromViz.tex
\newcommand{\nOne}{9}
\newcommand{\nTwo}{9}
\newcommand{\nThree}{3}
\newcommand{\nP}{3}
\pgfmathsetmacro{\offset}{.33}
\newcommand{\mycolor}{red}
\newcommand{\x}{0}

\newcommand{\shadingleft}{
\begin{scope}[canvas is zx plane at y=(\nTwo-.5-\x),rotate=90,shift={(-\nOne+.5,-.5)}]
	\draw[draw=\mycolor,fill=\mycolor!25] (0,\nThree) rectangle (\nOne,0);
\end{scope}
\begin{scope}[canvas is yz plane at x=.5,rotate=-90,yscale=-1,shift={(-.5,-\nOne+.5+\x)}]
	\draw[draw=\mycolor,fill=\mycolor!25] (0,0) rectangle (\nThree,\nOne/\nP);
\end{scope}
\begin{scope}[canvas is yx plane at z=.5,yscale=-1,rotate=180,shift={(-\nTwo+\x+.5,.5)}]
	\draw[draw=\mycolor,fill=\mycolor!25] (0,0) rectangle (\nOne/\nP,-\nOne);
\end{scope}
}

\newcommand{\shadingright}{
\begin{scope}[canvas is yx plane at z=.5-\x,yscale=-1,rotate=180,shift={(.5,.5)}]
	\draw[draw=\mycolor,fill=\mycolor!25] (0,0) rectangle (-\nOne,\nThree);
\end{scope}
\begin{scope}[canvas is yz plane at x=.5+\nThree,rotate=-90,yscale=-1,shift={(-.5+\x,-\nOne+.5)}]
	\draw[draw=\mycolor,fill=\mycolor!25] (0,0) rectangle (\nThree/\nP,\nOne);
\end{scope}
\begin{scope}[canvas is zx plane at y=(\nTwo-.5),rotate=90,shift={(.5,-.5+\x)}]
	\draw[draw=\mycolor,fill=\mycolor!25] (0,0) rectangle (\nThree,\nThree/\nP);
\end{scope}
}

\newcommand{\shadingnoredist}{
\begin{scope}[canvas is yx plane at z=.5,yscale=-1,rotate=180,shift={(-\nTwo+\x+.5,.5)}]
	\draw[draw=\mycolor,fill=\mycolor!25] (0,\nThree) rectangle (\nOne/\nP,-\nOne);
\end{scope}
\begin{scope}[canvas is zx plane at y=(\nTwo-.5-\x),rotate=90,shift={(-\nOne+.5,-.5)}]
	\draw[draw=\mycolor,fill=\mycolor!25] (0,\nThree) rectangle (\nOne+\nThree,0);
\end{scope}
\begin{scope}[canvas is yz plane at x=.5+\nThree,rotate=-90,yscale=-1,shift={(-.5,-\nOne+.5+\x)}]
	\draw[draw=\mycolor,fill=\mycolor!25] (0,0) rectangle (\nThree,\nOne/\nP);
\end{scope}
}

\newcommand{\compprism}{
\node[shift={(.5,-.5,.5)},scale=2] at (-\nOne/2,-0.5,0) {$n$};
\node[shift={(.5,-.5,.5)},scale=2] at (\nThree/2,-0.5,0) {$r$};
\node[shift={(.5,-.5,.5)},scale=2] at (\nThree+0.5/2,-0.5/2,-\nThree/2) {$r$};
\node[shift={(.5,-.5,.5)},scale=2] at (-\nOne-0.5,\nTwo/2,0) {$n$};

\begin{scope}[canvas is yz plane at x=.5,rotate=-90,yscale=-1,shift={(-.5,-\nTwo+.5)}]
	\draw[black,xscale=\nThree/1,yscale=\nTwo/1] (0,0) grid (1,1);
	\node[yscale=-1,scale=2] at (\nThree/2,\nTwo/2) {\Huge$\M{B}$};
\end{scope}
		
\begin{scope}[canvas is yz plane at x=.5-\nOne,rotate=-90,yscale=-1,shift={(-.5,-\nTwo+.5)}]
	\draw[black,dashed,xscale=\nThree/1,yscale=\nTwo/1] (0,0) grid (1,1);
\end{scope}

\begin{scope}[canvas is yx plane at z=.5,yscale=-1,rotate=180,shift={(-\nTwo+.5,-\nOne+.5)}]
	\draw[black,xscale=\nTwo/1,yscale=\nOne/1] (0,0) grid (1,1);
	\node[rotate=90,scale=2] at (\nTwo/2,\nOne/2) {\Huge$\M{A}$};
\end{scope}

\begin{scope}[canvas is yx plane at z=.5-\nThree,yscale=-1,rotate=180,shift={(-\nTwo+.5,-\nOne+.5)}]
	\draw[black,dashed,xscale=\nTwo/1,yscale=\nOne/1] (0,0) grid (1,1);
\end{scope}

\begin{scope}[canvas is zx plane at y=(\nTwo-.5),rotate=90,shift={(-\nOne+.5,-.5)}]
	\draw[black,xscale=\nOne/1,yscale=\nThree/1] (0,0) grid (1,1);
	\node[rotate=90,scale=2] at (\nOne/2,\nThree/2) {\Huge$\M{\Omega}$};
\end{scope}
		
\begin{scope}[canvas is zx plane at y=(\nTwo-.5),rotate=90,shift={(.5,-.5)}]
	\draw[black,xscale=\nThree/1,yscale=\nThree/1] (0,0) grid (1,1);
	\node[rotate=90,scale=2] at (\nThree/2,\nThree/2) {\Huge$\M{C}$};
\end{scope}
		
\begin{scope}[canvas is yx plane at z=.5,yscale=-1,rotate=180,shift={(.5,.5)}]
	\draw[black,xscale=\nTwo/1,yscale=\nThree/1] (0,0) grid (-1,1);
	\node[rotate=90,scale=2] at (-\nTwo/2,\nThree/2) {\Huge$\M{\Omega}$};
\end{scope}
		
\begin{scope}[canvas is yz plane at x=.5+\nThree,rotate=-90,yscale=-1,shift={(-.5,-\nTwo+.5)}]
	\draw[black,xscale=\nThree/1,yscale=\nTwo/1] (0,0) grid (1,1);
	\node[yscale=-1,scale=2] at (\nThree/2,\nTwo/2) {\Huge$\M{B}$};
\end{scope}
		
\begin{scope}[canvas is yx plane at z=.5-\nThree,yscale=-1,rotate=180,shift={(-\nTwo+.5,.5)}]
	\draw[black,dashed,xscale=\nTwo/1,yscale=\nThree/1] (0,0) grid (1,1);
\end{scope}
}

\begin{tikzpicture}[every node/.append style={transform shape},scale=.3]
		
	\renewcommand{\x}{6}
	\renewcommand{\mycolor}{blue}
	\shadingleft
		
	\renewcommand{\x}{3}
	\renewcommand{\mycolor}{green}
	\shadingleft
			
	\renewcommand{\x}{0}
	\renewcommand{\mycolor}{red}
	\shadingleft	
			
	\renewcommand{\x}{2}
	\renewcommand{\mycolor}{blue}
	\shadingright
			
	\renewcommand{\x}{1}
	\renewcommand{\mycolor}{green}
	\shadingright
			
	\renewcommand{\x}{0}
	\renewcommand{\mycolor}{red}
	\shadingright
		
	\compprism
	
	\begin{scope}[shift={(\nOne+\nThree+\nThree,0)}]
		\renewcommand{\x}{6}
		\renewcommand{\mycolor}{blue}
		\shadingnoredist
		
		\renewcommand{\x}{3}
		\renewcommand{\mycolor}{green}
		\shadingnoredist
			
		\renewcommand{\x}{0}
		\renewcommand{\mycolor}{red}
		\shadingnoredist
	
		\compprism
	\end{scope}

\end{tikzpicture}

%% file: experiments.tex

\subsection{Implementation and Experimental Setting}

\textbf{Implementation.}
We implement two above mentioned algorithms (\emph{Redist} and \emph{No-Redist}) for both CPU-only and GPU equipped systems using C++ and CUDA.
The compilers and libraries used for our implementation are given in \cref{tab:system_info}.
We perform all our experiments in double precision.
All of our experiments use the Intel MKL library for local CPU performance.
We also benchmarked single-threaded and multi-threaded matrix multiplication using Cray's LibSci library and AMD's AOCL library and observed comparable performance across all three.
In order to ease the adoption of our algorithms by practitioners, we also implement our algorithms using Python-based libraries for CPU-only systems.
For our Python implementation we use Numpy version 1.26.3 for local computation and mpi4py version 3.1.5 to perform interprocessor communication.

\begin{table}[!b]
	\centering
	\caption{Overview of the evaluation platform.}
	\begin{tabular}{r l l}
		\toprule
		\multicolumn{3}{c}{\textbf{NERSC Perlmutter GPU Partition}} \\
		\toprule
		Processor & \multicolumn{2}{c}{1 $\times$ AMD EPYC 7763 per node} \\
		\# NUMA domains & \multicolumn{2}{c}{4} \\
		\# Cores & \multicolumn{2}{c}{64 per node, 16 per NUMA domain }\\
		\# Hyperthreads & \multicolumn{2}{c}{2 per core, 32 per NUMA domain }\\
		\# GPUs & \multicolumn{2}{c}{4 $\times$ NVIDIA A100 per node }\\
		Memory & \multicolumn{2}{c}{256 GB DDR4 per node, 40GB per GPU }\\
		\toprule
		& \textbf{CPU-only} & \textbf{GPU} \\
		\toprule
		Compiler & Intel C++  & NVIDIA CUDA C++ \\
		& version 2024.1.0 & version 24.5-1 \\
		BLAS and PRNG & MKL 2024.1 & CUDA toolkit 12.4 \\
		MPI & \multicolumn{2}{c}{CRAY MPICH 8.1.30 (CUDA-aware for GPUs)} \\
		\bottomrule
	\end{tabular}

	\label{tab:system_info}
\end{table}

\textbf{Evaluation platform.}
We evaluate our implementation on the GPU partition of NERSC Perlmutter using 32 compute nodes (2048 cores, 128 GPUs).
The details of each compute node are given in \cref{tab:system_info}.
For both GPU and CPU-only evaluations, we run 4 MPI processes per compute node, assigning each process to a distinct GPU or CPU NUMA domain.
In the CPU-only runs, each process utilizes 32 threads for multithreading - 16 physical cores within its own NUMA domain and two logical threads per core. We report the average execution time over 10 runs for each experiment.

\begin{figure}[!t]
	\centering
	\input{plots/cpp_python.tex}
	\caption{Performance of our 3D-distributed memory matrix multiplication implemented in C++ and Python, multiplying two $50k \times 50k$ double precision matrices.
		We also include SLATE~\cite{Gates2019SLATE}, a widely-used distributed linear algebra library, as a reference for our C++ implementation.}
	\label{fig:matmul-cpp-v-py}
\end{figure}

\textbf{Parallel performance comparison of C++ vs Python.}
To compare the parallel performance between C++ and Python we perform a CPU-only experiment with distributed 3D matrix multiplication by multiplying two $50k \times 50k$ double precision matrices.
The multiplication operation involved gathering parts of both matrices, performing multiplication locally and then reducing the partial result by MPI reduce-scatter.
The result of the experiment is presented in \cref{fig:matmul-cpp-v-py}.
We observe that the performance of local matrix multiplication is comparable across the two implementations as both of them utilize MKL for BLAS operations.
However, even though both of the implementations use the same underlying MPI library, the communication costs were significantly higher for the Python implementation.
Because we could not explain the performance degradation and the C++ implementation is more efficient, in the following subsections, we omit our Python implementation for performance evaluation.
We also compare our C++ implementation against SLATE~\cite{Gates2019SLATE}, a widely-used distributed linear algebra library, on the same hardware.
Our implementation performs marginally faster than SLATE, validating the efficiency of our baseline implementation.

\subsection{Datasets}

\textbf{Metabarcoding genetic dissimilarity matrix.}
For an experiment focused on matrix multiplication with a random matrix, we consider a genetic dissimilarity matrix arising in a metabarcoding application, where the dissimilarity matrix is a $10^6 \times 10^6$ symmetric matrix~\cite{DS:NKTRHO_2023}.
In this matrix, each entry represents genetic distance between a pair of diatoms~\cite{agullo2022task}.
Following previous work, we use $r=1000$ in sketching this data \cite{ABC+23}.

\textbf{CIFAR-10 kernel matrix.}
In order to compare Nystr\"{o}m algorithms across meaningful matrix dimensions, we use the CIFAR-10 dataset \cite{K09}, a benchmark dataset for image classification/object recognition that consists of 50,000 images of $32\times 32$ RGB pixels.
The dataset has been used previously to study Nystr\"{o}m aproximation~\cite{frangella2023randomized}.
We consider the data as a $50,000 \times 3072$ input matrix and compute a symmetric kernel matrix using two different kernel functions.
To generate a kernel matrix with known rank, we use a linear kernel: $A_{ij} = K(\V{x}_i,\V{x}_j) = \V{x}_i^\Tra\V{x}_j$, where $\V{x}_i$ and $\V{x}_j$ are rows of the input matrix.
To generate a more realistic kernel with unknown rank, we also use the radial basis function (RBF) kernel: $A_{ij} = K(\V{x}_i,\V{x}_j) = \text{exp}\left(-\frac{\|\V{x}_i - \V{x}_j\|^2}{2\sigma^2}\right)$, for some parameter $\sigma$.
We compute Nystr\"{o}m approximation with various ranks and report the relative approximation error in \cref{tab:cifar10-error}.
In computing the explicit reconstruction of the approximation, we construct the pseudoinverse of the core matrix using a numerical tolerance of \texttt{1e-12}.

From the data, we see that the linear kernel is approximated to high accuracy for $r=5000$, which exceeds the true rank of the data, but it is also well approximated by ranks as low as $r=500$.
The RBF kernel yields higher numerical ranks, depending on the parameter $\sigma$, and may need ranks of at least $r=5000$ for a meaningful approximation.
Based on this dataset, we target the input matrix dimension $n=50000$ and ranks $r=500$ and $r=5000$ in our experiments.

\begin{table}[!t]
    \centering
	\caption{Approximation error of the CIFAR‑10 kernel matrix for different kernel functions and different ranks.}
	\begin{tabular}{| c | c c c |}
	  \hline
	  Kernel & r=500 & r=2500 & r=5000 \\
	  \hline
	  Linear & \texttt{5.8e-04} & \texttt{1.6e-05} & \texttt{6.3e-06} \\      
	  RBF $\sigma=\|\M{X}\|/\sqrt{n}$ & \texttt{1.0e-03} & \texttt{2.0e-04}  & \texttt{9.3e-05} \\     
	  RBF $\sigma=1$ & \texttt{9.9e-01} & \texttt{9.7e-01}  & \texttt{9.5e-01} \\     \hline
	\end{tabular}
	\label{tab:cifar10-error}
\end{table}

\subsection{Communicate vs Redundantly Generate $\M{\Omega}$}
\label{subsec:gen-vs-comm}

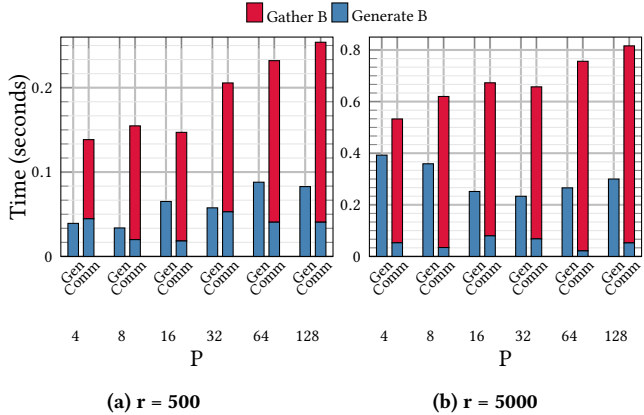
\begin{figure}[t]
    \centering

    \begin{subfigure}[t]{0.48\columnwidth}   
        \centering
        \input{plots/matmul_cpu_500.tex}
        \caption*{(a) r = 500}
    \end{subfigure}
    \hfill
    \begin{subfigure}[t]{0.48\columnwidth}   
        \centering
        \input{plots/matmul_cpu_5000.tex}
        \caption*{(b) r = 5000}
    \end{subfigure}

    \caption{Comparison between generating $\M{\Omega}$ redundantly vs communicating $\M{\Omega}$ in CPU-only systems for different values of $r$ to compute $\M{B}=\M{A}\M{\Omega}$ where $\M{A}$ is a CIFAR10 kernel matrix with dimensions $50k \times 50k$.}
    \label{fig:matmul1-gen-v-comm}
\end{figure}

The multiplication $\M{B}=\M{A} \M{\Omega}$ is the most expensive part of both of our algorithms as it involves the large input data. 
We use the row-wise 1D algorithm in both cases because the algorithm requires communicating only one matrix, although each processor requires the entire $\M{\Omega}$.
Because it is a random matrix, it is possible to generate the entire $\M{\Omega}$ with a shared seed on each process as opposed to generating parts and all-gathering it from other processors.
In this way, no inter-processor communication is required.

\Cref{fig:matmul1-gen-v-comm} depicts the experimental results comparing the two approaches to instantiating $\M{\Omega}$ on all $P$ processors using CPUs.
In this experiment, we compute $\M{B}=\M{A} \M{\Omega}$ where $\M{A}$ is a CIFAR10 kernel matrix and $\M{\Omega}$ has 500 and 5000 columns, but we show only the costs of accumulating $\M{\Omega}$.
Here we generate $\M{\Omega}$ as a uniform random matrix using the counter-based pseudorandom number generation algorithm Philox \cite{salmon2011parallel} available on both MKL and cuRAND.
We show only the data for the CPU-only experiments as we find similar qualitative behavior on GPUs.
We see that generating $\M{\Omega}$ redundantly takes less time than communicating it through out all experiments, even at low process count.
While the communication time increases with the increase of number of processes, the random number generation remains constant as expected.
Based on the evidence from this experiment, we generate $\M{\Omega}$ redundantly in both of our algorithms using the same generators used in this evaluation.

\subsection{Matrix Multiplication with Random Matrix}
\label{subsec:matmul1gen-scaling}
 
  \begin{figure}[b]
	\centering
	\input{plots/matmul1gen_cpu_scaling.tex}
	\caption{
		Total runtime of multiplying $10^6 \times 10^6$ genetic dissimilarity matrix with a random matrix with $10^3$ columns.
	}
	\label{fig:matmul1gen-scaling}
	\end{figure}
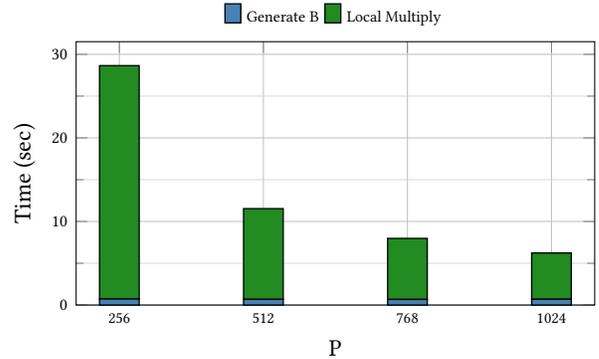

We now consider a strong scaling experiment of sketching the metabarcoding data, a large $10^6 \times 10^6$ symmetric matrix.
Given our observation in \cref{subsec:gen-vs-comm}, we redundantly generate the random matrix rather than communicating it in this experiment.
\Cref{fig:matmul1gen-scaling} presents the timing results of the experiment.
We start the experiment using 64 nodes (256 MPI processes), as that is the memory required to store the data in memory.
We note that these experiments are CPU-only because the input matrix does not fit into GPU memory for a reasonable number of GPUs, making it impractical to benchmark effectively.

As expected, we observe perfect strong scaling.
When the number of rows of $\M{A}$ is large, the 1D algorithm performs no inter-process communication: each process independently generates the entire random matrix $\M{\Omega}$ and performs local matrix multiplication. 
Thus, we expect perfect strong scaling as long as neither the generation of $\M{\Omega}$ nor degraded local multiplication performance due to poor aspect ratio becomes a bottleneck.
In fact, we observe greater than $2\times$ speedup from $P=256$ to $P=512$, which is due to the local output matrix fitting into cache for $P\geq 512$.
For input matrices that are short and wide, we expect that switching to the 2D algorithm will become necessary for smaller ranges of $P$.

\subsection{Nystr\"{o}m Computation: CPU vs GPU}
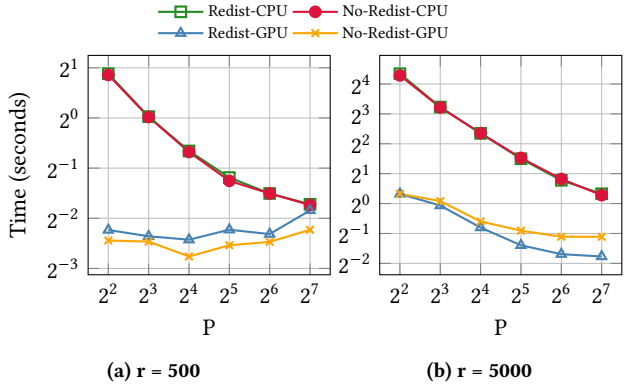
\begin{figure}[t]
    \centering

    \begin{subfigure}{0.48\columnwidth}
        \centering
        \input{plots/lineplot500.tex}
        \caption*{(a) r = 500}
    \end{subfigure}
    \hfill
    \begin{subfigure}{0.48\columnwidth}
        \centering
        \input{plots/lineplot5000.tex}
        \caption*{(b) r = 5000}
    \end{subfigure}
\hfill

    \caption{Total runtime of Nystr\"{o}m computation with \emph{Redist} and \emph{No-Redist} algorithms on both CPU-only and GPU equipped systems when approximating CIFAR10 kernel matrix to different ranks.}
    \label{fig:nystrom-allfour}
\end{figure}

\begin{figure}[t]
    \centering

    \begin{subfigure}{0.475\columnwidth}
        \centering
        \input{plots/nystrom_cpu_5000_all.tex}
        \caption*{(a) CPU}
    \end{subfigure}
    \hfill
    \begin{subfigure}{0.475\columnwidth}
        \centering
        \input{plots/nystrom_gpu_5000_all.tex}
        \caption*{(b) GPU}
    \end{subfigure}
    \caption{
        Runtime breakdown of Nystr\"{o}m computation with \emph{Redist} and \emph{No-Redist} algorithms on both CPU-only and GPU equipped systems when approximating CIFAR10 kernel matrix to $r=5000$.
        \emph{No-Redist} uses a reduce-scatter collective to perform communication while \emph{Redist} uses an all-to-all.
        Our storage format is column major order, therefore there is an unpacking step after all-to-all.
    }
    \label{fig:nystrom-breakdown-r5000}
\end{figure}
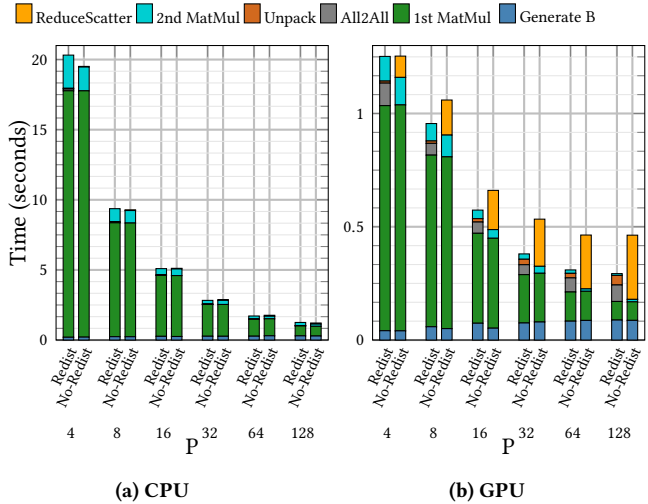

We benchmark the performance to approximate the CIFAR10 kernel matrices to rank 500 and 5000 using our algorithms.
We present the performance and scalability of both of our algorithms on both systems (CPU-only and GPU) in \cref{fig:nystrom-allfour}.
Our observations from this evaluation are as follows -

\textbf{GPU implementations are faster}.
We observe that for both smaller and larger $r$, GPU implementations are about an order of magnitude faster than the CPU implementation counterpart.
This is expected due to faster matrix multiplication on GPUs, which is the dominant cost as shown in \cref{fig:nystrom-breakdown-r5000}.

\textbf{CPU-only implementations scale better}.
Even though GPU implementations are faster, CPU-only implementations scale better than both algorithms on GPUs.
We observe from \cref{fig:nystrom-breakdown-r5000} that the total runtime of the CPU-only implementation is dominated by the local matrix multiplication longer.

\textbf{\emph{Redist} algorithm scales better}.
If we look at the performance of the GPU implementation, we see that the \emph{Redist} algorithm scales better.
This is expected, as the bandwidth cost of the redistribution (using all-to-all) is $O(nr/P)$.
The communication cost of the \emph{No-Redist} algorithm is that of a reduce-scatter with bandwidth cost $O(r^2)$.
The CPU implementations perform similarly in overall time, but this is because the time is more heavily dominated by local matrix multiplications.
We explore the empirical communication costs further in \cref{subsec:communication-comparison}.

\subsection{\scalebox{0.975}{Communication Cost of Redist vs No-Redist}}
\label{subsec:communication-comparison}

\begin{figure}[!b]
	\centering
	\begin{subfigure}{0.475\columnwidth}
		\centering
		\input{plots/nystrom_cpu_500.tex}
		\caption{r = 500 \label{fig:nystrom-communication-cpu-500}}      
	\end{subfigure}
	\hfill
	\begin{subfigure}{0.475\columnwidth}
		\centering
		\input{plots/nystrom_cpu_5000.tex}
		\caption{r = 5000 \label{fig:nystrom-communication-cpu-5000}}
	\end{subfigure}
	
	\caption{Communication cost of the two algorithms on CPU-only systems.}
	\label{fig:nystrom-communication-cpu}
\end{figure}
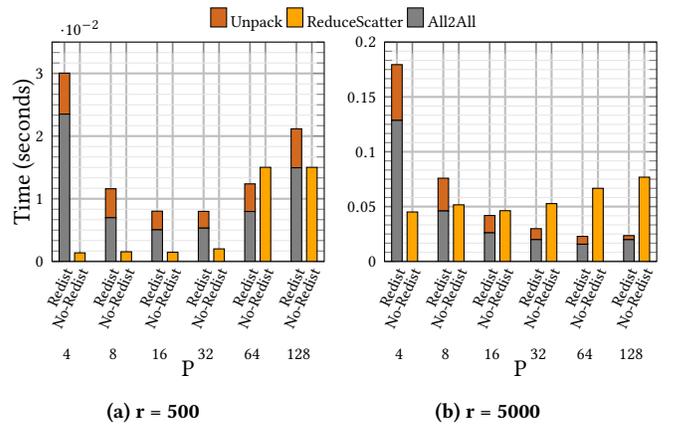

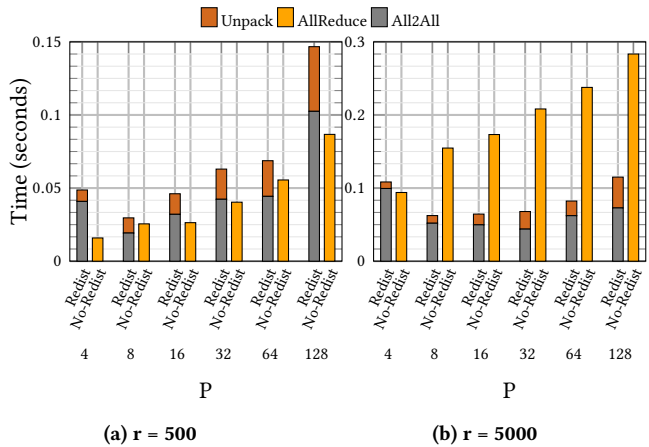
\begin{figure}[!t]
	\centering
	
	\begin{subfigure}{0.475\columnwidth}
		\centering
		\input{plots/nystrom_gpu_500.tex}
		\caption{r = 500 \label{fig:nystrom-communication-gpu-500}}
	\end{subfigure}
	\hfill
	\begin{subfigure}{0.475\columnwidth}
		\centering
		\input{plots/nystrom_gpu_5000.tex}
		\caption{r = 5000 \label{fig:nystrom-communication-gpu-5000}}
	\end{subfigure}
	
	\caption{Communication cost of the two algorithms on GPU-equipped systems.}
	\label{fig:nystrom-communication-gpu}
\end{figure}

Both of our algorithms have same amount of computation but the main difference comes in communication.
Hence, we focus only on evaluating the communication cost for these two algorithms when approximating the CIFAR10 kernel matrix using $r=500$ and $r=5000$.
The communication for the \emph{No-Redist} algorithm involves reduce-scatter of $O(r^2)$ input entries while for the \emph{Redist} algorithm involves all-to-all of $O(nr/P)$ entries per processor.
We use column major storage for the matrices, hence, \emph{Redist} algorithm requires unpacking after all-to-all.
The bandwidth bound communication cost is $O(r^2)$ for \emph{No-Redist} vs $O(nr/P)$ for \emph{Redist}.
Consequently, we expect the communication cost to scale with $P$ for the \emph{Redist} algorithm but not for the \emph{No-Redist} algorithm.
In Fig.~\ref{fig:nystrom-communication-cpu} we show the comparison for the CPU-only implementation and in Fig.~\ref{fig:nystrom-communication-gpu} we present it for the GPU implementations.
Our observations are as follows.

\textbf{Experiment matches theory for the CPU-only implementations}.
From \cref{fig:nystrom-communication-cpu-5000}, for $r=5000$, we observe that the communication cost starts higher for \emph{Redist} algorithm than for the \emph{No-Redist} algorithm in the case of small $P$, but it switches and becomes cheaper in the case of larger $P$.
The switch happens roughly around $P=\frac{n}{r}$, which matches exactly with the expectation.
For $r=500$ (\cref{fig:nystrom-communication-cpu-500}) we observe communication cost of the \emph{Redist} algorithm to decrease and increase again.
We attribute that to the communication becoming latency bound or perhaps a change in underlying algorithm.

\textbf{Communication in GPU implementation does not follow the same pattern}.
From \cref{fig:nystrom-communication-gpu} we observe that communication for \emph{No-Redist} increases at a rate close to $\log{p}$ and for \emph{Redist} algorithm the communication does not scale similar to the CPU-only experiments.
We suspect one reason for that to be due to differences within the CUDA-aware MPI implementation.
In particular, we initially observed that reduce scatter scaled very poorly within our GPU implementation and in a microbenchmark.
We alleviated this problem by replacing the reduce-scatter with allreduce for the GPU implementation which gives more reasonable performance despite its higher theoretical costs.

%% file: plots/cpp_python.tex

\pgfplotstableread[col sep=comma]{plots/data/cpp_python.csv}\datafile

\begin{tikzpicture}
\begin{axis}[\cppVSpython
    height=0.25\textwidth,
    after end axis/.code={ 
      \node(A14) at (axis cs:Algo-\impla-\proca,0) {};
      \node(A24) at (axis cs:Algo-\implb-\proca,0) {};
      \node(A34) at (axis cs:Algo-\implc-\proca,0) {};
      \node(A18) at (axis cs:Algo-\impla-\procb,0) {};
      \node(A28) at (axis cs:Algo-\implb-\procb,0) {};
      \node(A38) at (axis cs:Algo-\implc-\procb,0) {};
      \node(A116) at (axis cs:Algo-\impla-\procc,0) {};
      \node(A216) at (axis cs:Algo-\implb-\procc,0) {};
      \node(A316) at (axis cs:Algo-\implc-\procc,0) {};
      \node(A132) at (axis cs:Algo-\impla-\procd,0) {};
      \node(A232) at (axis cs:Algo-\implb-\procd,0) {};
      \node(A332) at (axis cs:Algo-\implc-\procd,0) {};
      \node(A164) at (axis cs:Algo-\impla-\proce,0) {};
      \node(A264) at (axis cs:Algo-\implb-\proce,0) {};
      \node(A364) at (axis cs:Algo-\implc-\proce,0) {};
      \node(A1128) at (axis cs:Algo-\impla-\procf,0) {};
      \node(A2128) at (axis cs:Algo-\implb-\procf,0) {};
      \node(A3128) at (axis cs:Algo-\implc-\procf,0) {};
      \node[yshift=-30, xshift=-5] at (A24) {\scriptsize $4$};
      \node[yshift=-30, xshift=-5] at (A28) {\scriptsize $8$};
      \node[yshift=-30, xshift=-5] at (A216) {\scriptsize $16$};
      \node[yshift=-30, xshift=-5] at (A232) {\scriptsize $32$};
      \node[yshift=-30, xshift=-5] at (A264) {\scriptsize $64$};
      \node[yshift=-30, xshift=-5] at (A2128) {\scriptsize $128$};
    },
    x tick label style={font=\scriptsize},
    y tick label style={font=\scriptsize,/pgf/number format/fixed,/pgf/number format/precision=2},
    ymajorgrids
    ]
    \pgfplotsset{cycle list={black, fill=color7 \\ black, fill=color2 \\ black, fill=color1 \\ black, fill=color4 \\ black, fill=color5 \\}}

    \addplot table[x=alg, y expr=(\thisrow{gatherATime})] {\datafile};
    \addplot table[x=alg, y expr=(\thisrow{gatherBTime})] {\datafile};
    \addplot table[x=alg, y expr=(\thisrow{localMultiplyTime})] {\datafile};
    \addplot table[x=alg, y expr=(\thisrow{scatterReduceTime})] {\datafile};
    \addplot table[x=alg, y expr=(\thisrow{slateTotalTime})] {\datafile};
    \iflegend
    	\legend{Gather A, Gather B, Local Multiply, Reduce Scatter, SLATE Total};
    \fi
  \end{axis}
\end{tikzpicture}


%% file: plots/matmul_cpu_500.tex
\begin{tikzpicture}
     \pgfplotstableread[col sep=comma]{plots/data/matmul_cpu_500.csv}\datafile
	
     \begin{axis}[\matmulCPUrankFiveH
      after end axis/.code={ 
      \node(A14) at (axis cs:Algo-\gen-\proca,0) {};
      \node(A24) at (axis cs:Algo-\comm-\proca,0) {};
      \node(A18) at (axis cs:Algo-\gen-\procb,0) {};
      \node(A28) at (axis cs:Algo-\comm-\procb,0) {};
      \node(A116) at (axis cs:Algo-\gen-\procc,0) {};
      \node(A216) at (axis cs:Algo-\comm-\procc,0) {};
      \node(A132) at (axis cs:Algo-\gen-\procd,0) {};
      \node(A232) at (axis cs:Algo-\comm-\procd,0) {};
      \node(A164) at (axis cs:Algo-\gen-\proce,0) {};
      \node(A264) at (axis cs:Algo-\comm-\proce,0) {};
      \node(A1128) at (axis cs:Algo-\gen-\procf,0) {};
      \node(A2128) at (axis cs:Algo-\comm-\procf,0) {};
      \node[yshift=-30, xshift=-5] at (A24) {\scriptsize $4$};
      \node[yshift=-30, xshift=-5] at (A28) {\scriptsize $8$};
      \node[yshift=-30, xshift=-5] at (A216) {\scriptsize $16$};
      \node[yshift=-30, xshift=-5] at (A232) {\scriptsize $32$};
      \node[yshift=-30, xshift=-5] at (A264) {\scriptsize $64$};
      \node[yshift=-30, xshift=-5] at (A2128) {\scriptsize $128$};
    },
    x tick label style={font=\scriptsize},
    y tick label style={font=\scriptsize,/pgf/number format/fixed,/pgf/number format/precision=2},
    grid=both,
    major grid style={line width=0.8pt, draw=gray!50},
     minor grid style={line width=0.4pt, draw=gray!20},
    minor tick num=5
    ]
    \pgfplotsset{cycle list={ black, fill=color3 \\ black, fill=color2  \\}}
    
     \addplot table[x=alg, y expr=(\thisrow{genBTime})] {\datafile};
    \addplot table[x=alg, y expr=(\thisrow{gatherBTime})] {\datafile};

    \iflegend
    	\legend{Generate $\M{\Omega}$, Gather $\M{\Omega}$};
    \fi
  \end{axis}
\end{tikzpicture}


%% file: plots/matmul_cpu_5000.tex
\pgfplotstableread[col sep=comma]{plots/data/matmul_cpu_5000.csv}\datafile

\begin{tikzpicture}
\begin{axis}[\matmulCPUrankFiveT
    after end axis/.code={ 
      \node(A14) at (axis cs:Algo-\gen-\proca,0) {};
      \node(A24) at (axis cs:Algo-\comm-\proca,0) {};
      \node(A18) at (axis cs:Algo-\gen-\procb,0) {};
      \node(A28) at (axis cs:Algo-\comm-\procb,0) {};
      \node(A116) at (axis cs:Algo-\gen-\procc,0) {};
      \node(A216) at (axis cs:Algo-\comm-\procc,0) {};
      \node(A132) at (axis cs:Algo-\gen-\procd,0) {};
      \node(A232) at (axis cs:Algo-\comm-\procd,0) {};
      \node(A164) at (axis cs:Algo-\gen-\proce,0) {};
      \node(A264) at (axis cs:Algo-\comm-\proce,0) {};
      \node(A1128) at (axis cs:Algo-\gen-\procf,0) {};
      \node(A2128) at (axis cs:Algo-\comm-\procf,0) {};
      \node[yshift=-30, xshift=-5] at (A24) {\scriptsize $4$};
      \node[yshift=-30, xshift=-5] at (A28) {\scriptsize $8$};
      \node[yshift=-30, xshift=-5] at (A216) {\scriptsize $16$};
      \node[yshift=-30, xshift=-5] at (A232) {\scriptsize $32$};
      \node[yshift=-30, xshift=-5] at (A264) {\scriptsize $64$};
      \node[yshift=-30, xshift=-5] at (A2128) {\scriptsize $128$};
    },
    x tick label style={font=\scriptsize},
    y tick label style={font=\scriptsize,/pgf/number format/fixed,/pgf/number format/precision=2},
        grid=both,
    major grid style={line width=0.8pt, draw=gray!50},
     minor grid style={line width=0.4pt, draw=gray!20},
    minor tick num=5
    ]
    \pgfplotsset{cycle list={ black, fill=color3 \\ black, fill=color2  \\}}
    
     \addplot table[x=alg, y expr=(\thisrow{genBTime})] {\datafile};
    \addplot table[x=alg, y expr=(\thisrow{gatherBTime})] {\datafile};

    \iflegend
    \fi
  \end{axis}

\end{tikzpicture}


%% file: plots/matmul1gen_cpu_scaling.tex

\pgfplotstableread[col sep=comma]{plots/data/matmul1gen_cpu_data.csv}\datafile

\begin{tikzpicture}
\begin{axis}[
  xlabel={P},
  ylabel={Time (sec)},
  ybar stacked,
  bar width=15pt,
  legend style={
	  at={(0.5,1.02)},
	  anchor=south,
	  legend columns=-1,
	  draw=none,
	  font=\scriptsize
  },
  symbolic x coords={256,512,768,1024},
  xtick=data,
  x tick label style={font=\scriptsize},
  y tick label style={font=\scriptsize,/pgf/number format/fixed,/pgf/number format/precision=2},
  ymin=0,
  width=\columnwidth,
  height=0.6\columnwidth,
  grid=both,
  grid style={line width=.1pt, draw=gray!30},
  major grid style={line width=.2pt,draw=gray!50},
  minor y tick num=1,
  ymajorgrids=true,
  yminorgrids=true,
]

\addplot[fill=color3, draw=black, line width=0.5pt]
  table[x=nprocs, y=gen_b_time] {\datafile};
\addlegendentry{Generate $\M{\Omega}$}

\addplot[fill=color1, draw=black, line width=0.5pt]
  table[x=nprocs, y=local_multiply_time] {\datafile};
\addlegendentry{Local Multiply}

\end{axis}

\end{tikzpicture}

%% file: plots/lineplot500.tex
\begin{tikzpicture}
    \begin{axis}[
        xlabel={P},
        ylabel={Time (seconds)},
        ymin=0,         grid=major,
        xmode=log,
	log basis x=2,
	ymode=log,
	log basis y=2,
	width=4.8cm, height=4.5cm,
	legend style={
        at={(0.25,1.15)},     
        anchor=west,         
    },
     legend columns=2,
      legend style={draw=none, cells={align=left}, nodes={scale=0.7}},
        xtick={2^2,2^3,2^4,2^5,2^6,2^7},
        ytick={2^-3,2^-2,2^-1, 2^0, 2^1},
    ]
        \pgfplotsset{cycle list={black, fill=color1 \\ black, fill=color2 \\ black, fill=color3 \\ black, fill=color4 \\ black, fill=color5 \\ black, fill=color9 \\}}
    \addplot[
        mark=square,
        color=color1,
        thick,
    ] table [x=proc, y=1DredCPU, col sep=comma] {plots/data/lineplot500.csv};
    \addlegendentry{Redist-CPU}
    
    \addplot[
        mark=*,
        color=color2,
        thick,
    ] table [x=proc, y=1DnoredCPU, col sep=comma] {plots/data/lineplot500.csv};
    \addlegendentry{No-Redist-CPU}
    
    \addplot[
        mark=triangle,
        color=color3,
        thick,
    ] table [x=proc, y=1DredGPU, col sep=comma] {plots/data/lineplot500.csv};
    \addlegendentry{Redist-GPU}
    
    \addplot[
        mark=x,
        color=color4,
        thick,
    ] table [x=proc, y=1DnoredGPU, col sep=comma] {plots/data/lineplot500.csv};
    \addlegendentry{No-Redist-GPU}
    
    \end{axis}
\end{tikzpicture}

%% file: plots/lineplot5000.tex
\begin{tikzpicture}
    \begin{axis}[
        xlabel={P},
        ymin=0,
        grid=major,
        xmode=log,
	log basis x=2,
	ymode=log,
	log basis y=2,
	width=4.8cm, height=4.5cm,
	 legend pos={north west},
  	legend columns=2,
  	legend style={draw=none, cells={align=left}, nodes={scale=0.7}},
        xtick={2^2,2^3,2^4,2^5,2^6,2^7},
        ytick={2^-2,2^-1,2^-0, 2^1, 2^2, 2^3, 2^4},
   ]

    \addplot[
        mark=square,
        color=color1,
        thick,
    ] table [x=proc, y=1DredCPU, col sep=comma] {plots/data/lineplot5000.csv};

    \addplot[
        mark=*,
        color=color2,
        thick,
    ] table [x=proc, y=1DnoredCPU, col sep=comma] {plots/data/lineplot5000.csv};
    
    \addplot[
        mark=triangle,
        color=color3,
        thick,
    ] table [x=proc, y=1DredGPU, col sep=comma] {plots/data/lineplot5000.csv};
    
    \addplot[
        mark=x,
        color=color4,
        thick,
    ] table [x=proc, y=1DnoredGPU, col sep=comma] {plots/data/lineplot5000.csv};
    
    \end{axis}
\end{tikzpicture}

%% file: plots/nystrom_cpu_5000_all.tex

\pgfplotstableread[col sep=comma]{plots/data/nystrom_cpu_5000.csv}\datafile

\begin{tikzpicture}
 \begin{axis}[\nystromCPUrankFiveTall
    after end axis/.code={ 
      \node(A14) at (axis cs:Algo-\algoa-\proca,0) {};
      \node(A24) at (axis cs:Algo-\algob-\proca,0) {};
      \node(A18) at (axis cs:Algo-\algoa-\procb,0) {};
      \node(A28) at (axis cs:Algo-\algob-\procb,0) {};
      \node(A116) at (axis cs:Algo-\algoa-\procc,0) {};
      \node(A216) at (axis cs:Algo-\algob-\procc,0) {};
      \node(A132) at (axis cs:Algo-\algoa-\procd,0) {};
      \node(A232) at (axis cs:Algo-\algob-\procd,0) {};
      \node(A164) at (axis cs:Algo-\algoa-\proce,0) {};
      \node(A264) at (axis cs:Algo-\algob-\proce,0) {};
      \node(A1128) at (axis cs:Algo-\algoa-\procf,0) {};
      \node(A2128) at (axis cs:Algo-\algob-\procf,0) {};
      \node[yshift=-35, xshift=-5] at (A24) {\scriptsize $4$};
      \node[yshift=-35, xshift=-5] at (A28) {\scriptsize $8$};
      \node[yshift=-35, xshift=-5] at (A216) {\scriptsize $16$};
      \node[yshift=-35, xshift=-5] at (A232) {\scriptsize $32$};
      \node[yshift=-35, xshift=-5] at (A264) {\scriptsize $64$};
      \node[yshift=-35, xshift=-5] at (A2128) {\scriptsize $128$};
    },
    x tick label style={font=\scriptsize},
    y tick label style={font=\scriptsize,/pgf/number format/fixed,/pgf/number format/precision=2},
        grid=both,
    major grid style={line width=0.8pt, draw=gray!50},
     minor grid style={line width=0.4pt, draw=gray!20},
    minor tick num=5
    ]
    \pgfplotsset{cycle list={black, fill=color3 \\ black, fill=color1 \\ black, fill=color5 \\ black, fill=color9 \\ black, fill=color8 \\ black, fill=color4 \\}}
    
    \addplot table[x=alg, y expr=(\thisrow{genOmegaTime})] {\datafile};
     \addplot table[x=alg, y expr=(\thisrow{dgemm1Time})] {\datafile};
    \addplot table[x=alg, y expr=(\thisrow{all2allTime})] {\datafile};
    \addplot table[x=alg, y expr=(\thisrow{unpackTime})] {\datafile};
     \addplot table[x=alg, y expr=(\thisrow{dgemm2Time})] {\datafile};
      \addplot table[x=alg, y expr=(\thisrow{reduceScatterTime})] {\datafile};
    
    \iflegend
    	\legend{Generate $\M{\Omega}$, 1st MatMul, All2All, Unpack, 2nd MatMul, ReduceScatter};
    \fi
  \end{axis}
\end{tikzpicture}


%% file: plots/nystrom_gpu_5000_all.tex

\pgfplotstableread[col sep=comma]{plots/data/nystrom_gpu_5000.csv}\datafile

\begin{tikzpicture}
\begin{axis}[\nystromGPUrankFiveTall
    after end axis/.code={ 
      \node(A14) at (axis cs:Algo-\algoa-\proca,0) {};
      \node(A24) at (axis cs:Algo-\algob-\proca,0) {};
      \node(A18) at (axis cs:Algo-\algoa-\procb,0) {};
      \node(A28) at (axis cs:Algo-\algob-\procb,0) {};
      \node(A116) at (axis cs:Algo-\algoa-\procc,0) {};
      \node(A216) at (axis cs:Algo-\algob-\procc,0) {};
      \node(A132) at (axis cs:Algo-\algoa-\procd,0) {};
      \node(A232) at (axis cs:Algo-\algob-\procd,0) {};
      \node(A164) at (axis cs:Algo-\algoa-\proce,0) {};
      \node(A264) at (axis cs:Algo-\algob-\proce,0) {};
      \node(A1128) at (axis cs:Algo-\algoa-\procf,0) {};
      \node(A2128) at (axis cs:Algo-\algob-\procf,0) {};
      \node[yshift=-35, xshift=-5] at (A24) {\scriptsize $4$};
      \node[yshift=-35, xshift=-5] at (A28) {\scriptsize $8$};
      \node[yshift=-35, xshift=-5] at (A216) {\scriptsize $16$};
      \node[yshift=-35, xshift=-5] at (A232) {\scriptsize $32$};
      \node[yshift=-35, xshift=-5] at (A264) {\scriptsize $64$};
      \node[yshift=-35, xshift=-5] at (A2128) {\scriptsize $128$};
    },
    x tick label style={font=\scriptsize},
    y tick label style={font=\scriptsize,/pgf/number format/fixed,/pgf/number format/precision=2},
        grid=both,
    major grid style={line width=0.8pt, draw=gray!50},
     minor grid style={line width=0.4pt, draw=gray!20},
    minor tick num=5
    ]
       \pgfplotsset{cycle list={black, fill=color3 \\ black, fill=color1 \\ black, fill=color5 \\ black, fill=color9 \\ black, fill=color8 \\ black, fill=color4 \\}}
    
    \addplot table[x=alg, y expr=(\thisrow{genOmegaTime})] {\datafile};
     \addplot table[x=alg, y expr=(\thisrow{dgemm1Time})] {\datafile};
    \addplot table[x=alg, y expr=(\thisrow{all2allTime})] {\datafile};
    \addplot table[x=alg, y expr=(\thisrow{unpackTime})] {\datafile};
     \addplot table[x=alg, y expr=(\thisrow{dgemm2Time})] {\datafile};
      \addplot table[x=alg, y expr=(\thisrow{reduceScatterTime})] {\datafile};
    
    \iflegend
    \fi
  \end{axis}
\end{tikzpicture}


%% file: plots/nystrom_cpu_500.tex

\pgfplotstableread[col sep=comma]{plots/data/nystrom_cpu_500.csv}\datafile

\begin{tikzpicture}
\begin{axis}[\nystromCPUrankFiveH
    after end axis/.code={ 
      \node(A14) at (axis cs:Algo-\algoa-\proca,0) {};
      \node(A24) at (axis cs:Algo-\algob-\proca,0) {};
      \node(A18) at (axis cs:Algo-\algoa-\procb,0) {};
      \node(A28) at (axis cs:Algo-\algob-\procb,0) {};
      \node(A116) at (axis cs:Algo-\algoa-\procc,0) {};
      \node(A216) at (axis cs:Algo-\algob-\procc,0) {};
      \node(A132) at (axis cs:Algo-\algoa-\procd,0) {};
      \node(A232) at (axis cs:Algo-\algob-\procd,0) {};
      \node(A164) at (axis cs:Algo-\algoa-\proce,0) {};
      \node(A264) at (axis cs:Algo-\algob-\proce,0) {};
      \node(A1128) at (axis cs:Algo-\algoa-\procf,0) {};
      \node(A2128) at (axis cs:Algo-\algob-\procf,0) {};
      \node[yshift=-35, xshift=-5] at (A24) {\scriptsize $4$};
      \node[yshift=-35, xshift=-5] at (A28) {\scriptsize $8$};
      \node[yshift=-35, xshift=-5] at (A216) {\scriptsize $16$};
      \node[yshift=-35, xshift=-5] at (A232) {\scriptsize $32$};
      \node[yshift=-35, xshift=-5] at (A264) {\scriptsize $64$};
      \node[yshift=-35, xshift=-5] at (A2128) {\scriptsize $128$};
    },
    x tick label style={font=\scriptsize},
    y tick label style={font=\scriptsize,/pgf/number format/fixed,/pgf/number format/precision=2},
        grid=both,
    major grid style={line width=0.8pt, draw=gray!50},
     minor grid style={line width=0.4pt, draw=gray!20},
    minor tick num=5
    ]
    \pgfplotsset{cycle list={black, fill=color5 \\ black, fill=color4 \\ black, fill=color9 \\ black, fill=color4 \\ black, fill=color6 \\ black, fill=color6 \\ black, fill=color7 \\ black, fill=color8 \\ black, fill=color9 \\}}
    
    \addplot table[x=alg, y expr=(\thisrow{all2allTime})] {\datafile};
    \addplot table[x=alg, y expr=(\thisrow{reduceScatterTime})] {\datafile};
    \addplot table[x=alg, y expr=(\thisrow{unpackTime})] {\datafile};
    \iflegend
    	\legend{All2All,ReduceScatter,Unpack};
    \fi
  \end{axis}

\end{tikzpicture}


%% file: plots/nystrom_cpu_5000.tex

\pgfplotstableread[col sep=comma]{plots/data/nystrom_cpu_5000.csv}\datafile

\begin{tikzpicture}
 \begin{axis}[\nystromCPUrankFiveT
    after end axis/.code={ 
      \node(A14) at (axis cs:Algo-\algoa-\proca,0) {};
      \node(A24) at (axis cs:Algo-\algob-\proca,0) {};
      \node(A18) at (axis cs:Algo-\algoa-\procb,0) {};
      \node(A28) at (axis cs:Algo-\algob-\procb,0) {};
      \node(A116) at (axis cs:Algo-\algoa-\procc,0) {};
      \node(A216) at (axis cs:Algo-\algob-\procc,0) {};
      \node(A132) at (axis cs:Algo-\algoa-\procd,0) {};
      \node(A232) at (axis cs:Algo-\algob-\procd,0) {};
      \node(A164) at (axis cs:Algo-\algoa-\proce,0) {};
      \node(A264) at (axis cs:Algo-\algob-\proce,0) {};
      \node(A1128) at (axis cs:Algo-\algoa-\procf,0) {};
      \node(A2128) at (axis cs:Algo-\algob-\procf,0) {};
      \node[yshift=-35, xshift=-5] at (A24) {\scriptsize $4$};
      \node[yshift=-35, xshift=-5] at (A28) {\scriptsize $8$};
      \node[yshift=-35, xshift=-5] at (A216) {\scriptsize $16$};
      \node[yshift=-35, xshift=-5] at (A232) {\scriptsize $32$};
      \node[yshift=-35, xshift=-5] at (A264) {\scriptsize $64$};
      \node[yshift=-35, xshift=-5] at (A2128) {\scriptsize $128$};
    },
    x tick label style={font=\scriptsize},
    y tick label style={font=\scriptsize,/pgf/number format/fixed,/pgf/number format/precision=2},
        grid=both,
    major grid style={line width=0.8pt, draw=gray!50},
     minor grid style={line width=0.4pt, draw=gray!20},
    minor tick num=5
    ]
    \pgfplotsset{cycle list={black, fill=color5 \\ black, fill=color4 \\ black, fill=color9 \\ black, fill=color4 \\ black, fill=color6 \\ black, fill=color6 \\ black, fill=color7 \\ black, fill=color8 \\ black, fill=color9 \\}}
    
    \addplot table[x=alg, y expr=(\thisrow{all2allTime})] {\datafile};
    \addplot table[x=alg, y expr=(\thisrow{reduceScatterTime})] {\datafile};
    \addplot table[x=alg, y expr=(\thisrow{unpackTime})] {\datafile};
    \iflegend
    \fi
  \end{axis}
\end{tikzpicture}


%% file: plots/nystrom_gpu_500.tex

\pgfplotstableread[col sep=comma]{plots/data/nystrom_gpu_500.csv}\datafile

\begin{tikzpicture}
\begin{axis}[\nystromGPUrankFiveH
    after end axis/.code={ 
      \node(A14) at (axis cs:Algo-\algoa-\proca,0) {};
      \node(A24) at (axis cs:Algo-\algob-\proca,0) {};
      \node(A18) at (axis cs:Algo-\algoa-\procb,0) {};
      \node(A28) at (axis cs:Algo-\algob-\procb,0) {};
      \node(A116) at (axis cs:Algo-\algoa-\procc,0) {};
      \node(A216) at (axis cs:Algo-\algob-\procc,0) {};
      \node(A132) at (axis cs:Algo-\algoa-\procd,0) {};
      \node(A232) at (axis cs:Algo-\algob-\procd,0) {};
      \node(A164) at (axis cs:Algo-\algoa-\proce,0) {};
      \node(A264) at (axis cs:Algo-\algob-\proce,0) {};
      \node(A1128) at (axis cs:Algo-\algoa-\procf,0) {};
      \node(A2128) at (axis cs:Algo-\algob-\procf,0) {};
      \node[yshift=-35, xshift=-5] at (A24) {\scriptsize $4$};
      \node[yshift=-35, xshift=-5] at (A28) {\scriptsize $8$};
      \node[yshift=-35, xshift=-5] at (A216) {\scriptsize $16$};
      \node[yshift=-35, xshift=-5] at (A232) {\scriptsize $32$};
      \node[yshift=-35, xshift=-5] at (A264) {\scriptsize $64$};
      \node[yshift=-35, xshift=-5] at (A2128) {\scriptsize $128$};
    },
    x tick label style={font=\scriptsize},
    y tick label style={font=\scriptsize,/pgf/number format/fixed,/pgf/number format/precision=2},
        grid=both,
    major grid style={line width=0.8pt, draw=gray!50},
     minor grid style={line width=0.4pt, draw=gray!20},
    minor tick num=5
    ]
    \pgfplotsset{cycle list={black, fill=color5 \\ black, fill=color4 \\ black, fill=color9 \\ black, fill=color4 \\ black, fill=color5 \\ black, fill=color9 \\}}
    
     \addplot table[x=alg, y expr=(\thisrow{all2allTime})] {\datafile};
     \addplot table[x=alg, y expr=(\thisrow{reduceScatterTime})] {\datafile};
    \addplot table[x=alg, y expr=(\thisrow{unpackTime})] {\datafile};

    \iflegend
    	\legend{ All2All, AllReduce, Unpack};
    \fi
  \end{axis}
\end{tikzpicture}


%% file: plots/nystrom_gpu_5000.tex

\pgfplotstableread[col sep=comma]{plots/data/nystrom_gpu_5000.csv}\datafile

\begin{tikzpicture}
\begin{axis}[\nystromGPUrankFiveT
    after end axis/.code={ 
      \node(A14) at (axis cs:Algo-\algoa-\proca,0) {};
      \node(A24) at (axis cs:Algo-\algob-\proca,0) {};
      \node(A18) at (axis cs:Algo-\algoa-\procb,0) {};
      \node(A28) at (axis cs:Algo-\algob-\procb,0) {};
      \node(A116) at (axis cs:Algo-\algoa-\procc,0) {};
      \node(A216) at (axis cs:Algo-\algob-\procc,0) {};
      \node(A132) at (axis cs:Algo-\algoa-\procd,0) {};
      \node(A232) at (axis cs:Algo-\algob-\procd,0) {};
      \node(A164) at (axis cs:Algo-\algoa-\proce,0) {};
      \node(A264) at (axis cs:Algo-\algob-\proce,0) {};
      \node(A1128) at (axis cs:Algo-\algoa-\procf,0) {};
      \node(A2128) at (axis cs:Algo-\algob-\procf,0) {};
      \node[yshift=-35, xshift=-5] at (A24) {\scriptsize $4$};
      \node[yshift=-35, xshift=-5] at (A28) {\scriptsize $8$};
      \node[yshift=-35, xshift=-5] at (A216) {\scriptsize $16$};
      \node[yshift=-35, xshift=-5] at (A232) {\scriptsize $32$};
      \node[yshift=-35, xshift=-5] at (A264) {\scriptsize $64$};
      \node[yshift=-35, xshift=-5] at (A2128) {\scriptsize $128$};
    },
    x tick label style={font=\scriptsize},
    y tick label style={font=\scriptsize,/pgf/number format/fixed,/pgf/number format/precision=2},
        grid=both,
    major grid style={line width=0.8pt, draw=gray!50},
     minor grid style={line width=0.4pt, draw=gray!20},
    minor tick num=5
    ]
        \pgfplotsset{cycle list={black, fill=color5 \\ black, fill=color4 \\ black, fill=color9 \\ black, fill=color4 \\ black, fill=color5 \\ black, fill=color9 \\}}
    
     \addplot table[x=alg, y expr=(\thisrow{all2allTime})] {\datafile};
     \addplot table[x=alg, y expr=(\thisrow{reduceScatterTime})] {\datafile};
    \addplot table[x=alg, y expr=(\thisrow{unpackTime})] {\datafile};
    \iflegend
    \fi
  \end{axis}
\end{tikzpicture}


%% file: refs.bib
@inproceedings{Gates2019SLATE,
  author    = {Gates, Mark and Kurzak, Jakub and Charara, Ali and YarKhan, Asim and Dongarra, Jack},
  title     = {{SLATE}: Design of a Modern Distributed and Accelerated Linear Algebra Library},
  booktitle = {Proceedings of the International Conference for High Performance Computing, Networking, Storage and Analysis},
  series    = {SC '19},
  year      = {2019},
  pages     = {1--18},
  publisher = {ACM},
  address   = {New York, NY, USA},
  doi       = {10.1145/3295500.3356223},
}

@inproceedings{HigBY25,
author = {Higgins, Andrew James and Boman, Erik and Yamazaki, Ichitaro},
title = {A High Performance GPU CountSketch Implementation and Its Application to Multisketching and Least Squares Problems},
year = {2025},
isbn = {9798400718717},
publisher = {Association for Computing Machinery},
address = {New York, NY, USA},
url = {https://doi.org/10.1145/3731599.3767544},
doi = {10.1145/3731599.3767544},
booktitle = {Proceedings of the SC '25 Workshops of the International Conference for High Performance Computing, Networking, Storage and Analysis},
pages = {1808–1815},
numpages = {8},
location = {
},
series = {SC Workshops '25}
}

@misc{CheNRSPK25,
      title={GPU-Parallelizable Randomized Sketch-and-Precondition for Linear Regression using Sparse Sign Sketches}, 
      author={Tyler Chen and Pradeep Niroula and Archan Ray and Pragna Subrahmanya and Marco Pistoia and Niraj Kumar},
      year={2025},
      eprint={2506.03070},
      archivePrefix={arXiv},
      primaryClass={cs.DS},
      url={https://arxiv.org/abs/2506.03070}, 
}

@article{BalG22,
author = {Balabanov, Oleg and Grigori, Laura},
title = {Randomized Gram--Schmidt Process with Application to GMRES},
journal = {SIAM Journal on Scientific Computing},
volume = {44},
number = {3},
pages = {A1450-A1474},
year = {2022},
doi = {10.1137/20M138870X},
URL = {https://doi.org/10.1137/20M138870X},
}

@InProceedings{BalBGL23,
  title = 	 {Block Subsampled Randomized Hadamard Transform for Nyström Approximation on Distributed Architectures},
  author =       {Balabanov, Oleg and Beaup\`{e}re, Matthias and Grigori, Laura and Lederer, Victor},
  booktitle = 	 {Proceedings of the 40th International Conference on Machine Learning},
  pages = 	 {1564--1576},
  year = 	 {2023},
  volume = 	 {202},
  series = 	 {Proceedings of Machine Learning Research},
  month = 	 {23--29 Jul},
  publisher =    {PMLR},
  url = 	 {https://proceedings.mlr.press/v202/balabanov23a.html},
}

@article{MarT20,
  title={Randomized numerical linear algebra: Foundations and algorithms},
  author={Martinsson, Per-Gunnar and Tropp, Joel A},
  journal={Acta Numerica},
  volume={29},
  pages={403--572},
  year={2020},
  publisher={Cambridge University Press},
  doi={10.1017/S0962492920000021},
}

@article{Nys930,
  title={{\"U}ber die praktisch aufl{\"o}sung von integralgleichungen mit anwendungen auf randwertaufgaben},
  author={Nystr{\"o}m, Evert J},
  year={1930},
  doi={10.1007/BF02547521},
}

@article{WilS00,
  title={Using the Nystr{\"o}m method to speed up kernel machines},
  author={Williams, Christopher and Seeger, Matthias},
  journal={Advances in neural information processing systems},
  volume={13},
  year={2000},
  url={https://dl.acm.org/doi/10.5555/3008751.3008847},
}

@inproceedings{Ballard:MTTKRP:IPDPS18,
	author = {Ballard, Grey and Knight, Nicholas and Rouse, Kathryn},
	booktitle = {2018 IEEE International Parallel and Distributed Processing Symposium (IPDPS)},
	doi = {10.1109/IPDPS.2018.00065},
	pages = {557-567},
	title = {Communication Lower Bounds for Matricized Tensor Times {Khatri-Rao} Product},
	year = {2018}}

@article{Thakur:CollectiveCommunications:2005,
	author = {R. Thakur and R. Rabenseifner and W. Gropp},
	doi = {10.1177/1094342005051521},
	journal = {Intl. J. High Perf. Comp. App.},
	number = {1},
	title = {Optimization of Collective Communication Operations in MPICH},
	volume = {19},
	year = {2005}}

@Book{BV04,
  title     = {Convex Optimization},
  publisher = {Cambridge University Press},
  year      = {2004},
  author    = {S. Boyd and L. Vandenberghe},
  isbn      = {0521833787},
  url       = {https://web.stanford.edu/~boyd/cvxbook/},
}

@article{Chan:CollectiveCommunications:2007,
	author = {Chan, E. and Heimlich, M. and Purkayastha, A. and van de Geijn, R.},
	doi = {https://doi.org/10.1002/cpe.1206},
	journal = {Conc. and Comp.: Prac. and Exper.},
	number = {13},
	title = {Collective Communication: Theory, Practice, and Experience},
	volume = {19},
	year = {2007}}

@InProceedings{BR20,
  author    = {G. Ballard and K. Rouse},
  title     = {General Memory-Independent Lower Bound for MTTKRP},
  booktitle = {SIAM PP},
  year      = {2020},
  pages     = {1-11},
  month     = jan,
  doi       = {10.1137/1.9781611976137.1},
  url       = {https://epubs.siam.org/doi/abs/10.1137/1.9781611976137.1},
}

@Article{ACS90,
  Title                    = {Communication Complexity of {PRAM}s},
  Author                   = {A. Aggarwal and A. K. Chandra and M. Snir},
  Journal                  = {Theor. Comp. Sci.},
  Year                     = {1990},
  Number                   = {1},
  Volume                   = {71},
  Doi                      = {10.1016/0304-3975(90)90188-N},
  ISSN                     = {0304-3975},
  Url                      = {http://www.sciencedirect.com/science/article/pii/030439759090188N}
}

@article{LW49,
	author = {L. H. Loomis and H. Whitney},
	journal = {Bulletin of the American Mathematical Society},
	number = {10},
	publisher = {American Mathematical Society},
	title = {An Inequality Related to the Isoperimetric Inequality},
	url = {https://doi.org/},
	volume = {55},
	year = {1949},
	doi = {10.1090/S0002-9904-1949-09320-5}
}

@article{IRONY:JPDC04,
	author = {Dror Irony and Sivan Toledo and Alexander Tiskin},
	doi = {https://doi.org/10.1016/j.jpdc.2004.03.021},
	issn = {0743-7315},
	journal = {Journal of Parallel and Distributed Computing},
	number = {9},
	pages = {1017-1026},
	title = {Communication Lower Bounds for Distributed-Memory Matrix Multiplication},
	url = {https://www.sciencedirect.com/science/article/pii/S0743731504000437},
	volume = {64},
	year = {2004}}

@techreport{Christ:EECS-2013-61,
	author = {Christ, Michael and Demmel, James and Knight, Nicholas and Scanlon, Thomas and Yelick, Katherine A.},
	institution = {EECS Department, University of California, Berkeley},
	month = {May},
	number = {UCB/EECS-2013-61},
	title = {Communication Lower Bounds and Optimal Algorithms for Programs That Reference Arrays - Part 1},
	url = {http://www2.eecs.berkeley.edu/Pubs/TechRpts/2013/EECS-2013-61.html},
	year = {2013}}

@InProceedings{HK81,
  Title                    = {{I/O} complexity: The Red-Blue Pebble Game},
  Author                   = {J. W. Hong and H. T. Kung},
  Booktitle                = {STOC 1981},
  Year                     = {1981},
  Doi                      = {http://doi.acm.org/10.1145/800076.802486},
}

@InProceedings{ABGKR22,
  author      = {Al Daas, H. and Ballard, G. and Grigori, L. and Kumar, S. and Rouse, K.},
  title       = {Tight Memory-Independent Parallel Matrix Multiplication Communication Lower Bounds},
  booktitle   = {SPAA 2022},
  year        = {2022},
  doi         = {10.1145/3490148.3538552},
  institution = {arXiv},
  isbn        = {978-1-4503-9146-7/22/07},
  techversion = {ABGKR22},
}

@ARTICLE{BCKUW97,
  author={Bruck, J. and Ching-Tien Ho and Kipnis, S. and Upfal, E. and Weathersby, D.},
  journal={IEEE Trans. on Par. and Dist. Sys.},
  title={Efficient Algorithms for All-to-All Communications in Multiport Message-Passing Systems},
  year={1997},
  volume={8},
  number={11},
  doi={10.1109/71.642949}
}

@inproceedings{ABC+23,
  title={On the Arithmetic Intensity of Distributed-Memory Dense Matrix Multiplication Involving a Symmetric Input Matrix (SYMM)},
  author={Agullo, Emmanuel and Buttari, Alfredo and Coulaud, Olivier and Eyraud-Dubois, Lionel and Faverge, Mathieu and Franc, Alain and Guermouche, Abdou and Jego, Antoine and Peressoni, Romain and Pruvost, Florent},
  booktitle={2023 IEEE International Parallel and Distributed Processing Symposium (IPDPS)},
  pages={357--367},
  year={2023},
  organization={IEEE},
  doi={10.1109/IPDPS54959.2023.00044}
}

@INPROCEEDINGS{TMBD-2024,
	author={Liang, Tianyu and Murray, Riley and Buluç, Aydın and Demmel, James},
	booktitle={2024 IEEE International Parallel and Distributed Processing Symposium (IPDPS)},
	title={Fast multiplication of random dense matrices with sparse matrices},
	year={2024},
	volume={},
	number={},
	pages={52-62},
	doi={10.1109/IPDPS57955.2024.00014}
}

@inproceedings{salmon2011parallel,
  title={Parallel random numbers: as easy as 1, 2, 3},
  author={Salmon, John K and Moraes, Mark A and Dror, Ron O and Shaw, David E},
  booktitle={Proceedings of 2011 international conference for high performance computing, networking, storage and analysis},
  pages={1--12},
  year={2011},
  doi = {10.1145/2063384.2063405},
}

@article{GM-2016,
	author = {Gittens, Alex and Mahoney, Michael W.},
	title = {Revisiting the Nystr\"{o}m method for improved large-scale machine learning},
	year = {2016},
	issue_date = {January 2016},
	publisher = {JMLR.org},
	volume = {17},
	number = {1},
	issn = {1532-4435},
	journal = {J. Mach. Learn. Res.},
	month = jan,
	pages = {3977–4041},
	numpages = {65},
  url = {https://dl.acm.org/doi/abs/10.5555/2946645.3007070},
}

@article{LLW-2023,
	author  = {Jian Li and Yong Liu and Weiping Wang},
	title   = {Optimal Convergence Rates for Distributed Nystroem Approximation},
	journal = {Journal of Machine Learning Research},
	year    = {2023},
	volume  = {24},
	number  = {141},
	pages   = {1--39},
	url     = {http://jmlr.org/papers/v24/21-1049.html}
}

@techreport{K09,
title = {Learning Multiple Layers of Features from Tiny Images},
author = {Krizhevsky, Alex}, 
institution =  {Computer Science Department, University of Toronto}, 
year = {2009},
url = {https://www.cs.toronto.edu/~kriz/learning-features-2009-TR.pdf},
}

@article{park2025accuracy,
  title={Accuracy and stability of CUR decompositions with oversampling},
  author={Park, Taejun and Nakatsukasa, Yuji},
  journal={SIAM Journal on Matrix Analysis and Applications},
  volume={46},
  number={1},
  pages={780--810},
  year={2025},
  publisher={SIAM},
  doi = {10.1137/24M1660346},
}

@misc{bucci2025numerical,
      title={Numerical Stability of the Nystr\"om Method},
      author={Alberto Bucci and Yuji Nakatsukasa and Taejun Park},
      year={2025},
      eprint={2511.15583},
      archivePrefix={arXiv},
      primaryClass={math.NA},
      url={https://arxiv.org/abs/2511.15583},
}

@article{cortinovis2025sublinear,
  title={A sublinear-time randomized algorithm for column and row subset selection based on strong rank-revealing QR factorizations},
  author={Cortinovis, Alice and Ying, Lexing},
  journal={SIAM Journal on Matrix Analysis and Applications},
  volume={46},
  number={1},
  pages={22--44},
  year={2025},
  publisher={SIAM},
  doi = {10.1137/24M164063X},
}

@article{cortinovis2026adaptive,
  title={Adaptive randomized pivoting for column subset selection, DEIM, and low-rank approximation},
  author={Cortinovis, Alice and Kressner, Daniel},
  journal={SIAM Journal on Matrix Analysis and Applications},
  volume={47},
  number={1},
  pages={25--47},
  year={2026},
  publisher={SIAM},
  doi = {10.1137/24M1719189},
}

@article{talwalkar2013large,
  title={Large-scale SVD and manifold learning},
  author={Talwalkar, Ameet and Kumar, Sanjiv and Mohri, Mehryar and Rowley, Henry},
  journal={The Journal of Machine Learning Research},
  volume={14},
  number={1},
  pages={3129--3152},
  year={2013},
  publisher={JMLR. org},
  url={http://jmlr.org/papers/v14/talwalkar13a.html}
}

@article{frangella2023randomized,
  title={Randomized nystr{\"o}m preconditioning},
  author={Frangella, Zachary and Tropp, Joel A and Udell, Madeleine},
  journal={SIAM Journal on Matrix Analysis and Applications},
  volume={44},
  number={2},
  pages={718--752},
  year={2023},
  publisher={SIAM},
  doi = {10.1137/21M1466244},
}

@techreport{agullo2022task,
  title={Task-based randomized singular value decomposition and multidimensional scaling},
  author={Agullo, Emmanuel and Coulaud, Olivier and Denis, Alexandre and Faverge, Mathieu and Franc, Alain and Frigerio, Jean-Marc and Furmento, Nathalie and Guilbaud, Adrien and Jeannot, Emmanuel and Peressoni, Romain and others},
  year={2022},
  number={RR-9482},
  institution={Inria Bordeaux-Sud Ouest; Inrae-BioGeCo},
  url={https://inria.hal.science/hal-03773985v2}
}

@data{DS:NKTRHO_2023,
author = {Franc, Alain and Frigerio, Jean-Marc and Chancerel, Emilie and Salin, Franck and Thérond, Sylvie and Rimet, Frédéric and Bouchez, Agnès},
publisher = {Recherche Data Gouv},
title = {{Reads and pairwise distances from 10 samples of diatoms in Geneva lake}},
year = {2023},
version = {V1},
doi = {10.57745/NKTRHO},
url = {https://doi.org/10.57745/NKTRHO}
}
